%% file: main.tex
\documentclass[a4paper,fontsize=10pt,DIV=13,numbers=noendperiod]{scrartcl}

\usepackage{booktabs}
\usepackage{subcaption} %
\usepackage{natbib}
\usepackage{savesym}
\usepackage{amsmath}
\usepackage{amsthm}
\usepackage{mathtools}
\usepackage{bussproofs}
\usepackage{enumerate}
\usepackage{amssymb}
\usepackage{stmaryrd}
\usepackage{paralist}
\usepackage{bbm}
\usepackage{hyperref}
\usepackage{galois}
\hypersetup{
    colorlinks,
    citecolor=brown,
    filecolor=black,
    linkcolor=black,
    urlcolor=brown
}
\usepackage[T1]{fontenc}
\usepackage{libertine}
\usepackage[libertine]{newtxmath}

\setkomafont{disposition}{\normalsize}
\setkomafont{section}{\large\sffamily}
\setkomafont{subsection}{\large\sffamily}
\setkomafont{paragraph}{\normalsize\rmfamily\itshape}

\savesymbol{comment}
\usepackage[nompar,nosign]{commenting}
\declareauthor{lo}{Luke}{red}
\authorcommand{lo}{comment}
\declareauthor{sr}{Steven}{blue}
\authorcommand{sr}{comment}
\declareauthor{tcb}{Toby}{olive}
\authorcommand{tcb}{comment}

\newtheorem{theorem}{Theorem}
\newtheorem{remark}{Remark}
\newtheorem{example}{Example}
\newtheorem{lemma}{Lemma}
\newtheorem{corollary}{Corollary}
\newtheorem{proposition}{Proposition}
\newif\ifsupp

\supptrue

\input{macros.tex}

\RedeclareSectionCommand[
  beforeskip=-\baselineskip,
  afterskip=.3\baselineskip]{section}
\RedeclareSectionCommand[
  beforeskip=-.75\baselineskip,
  afterskip=.3\baselineskip]{subsection}

\renewcommand{\cite}{\citep}

\makeatletter
\DeclareOldFontCommand{\it}{\normalfont\itshape}{\mathit}
\renewcommand\maketitle{
   \begin{center}
     {\Large\sffamily\@title}\\[4mm]
     \begin{center}
       \sffamily
       \begin{tabular}{c}
         Toby Cathcart Burn\\
         {\small University of Oxford}
       \end{tabular}
       \qquad
       \begin{tabular}{c}
         C.-H. Luke Ong\\
         {\small University of Oxford}
       \end{tabular}
       \qquad
       \begin{tabular}{c}
         Steven J. Ramsay\\
         {\small University of Oxford}
       \end{tabular}
     \end{center}
   \end{center}
}
\makeatother

\begin{document}

\title{Higher-Order Constrained Horn Clauses and Refinement Types}         %

\maketitle

\vspace{7mm}
\hrule
\begin{abstract}
  \small
\input{abstract.tex}
\end{abstract}
\hrule
\vspace{5mm}

\input{introduction.tex}
\input{preliminaries.tex}
\input{hornsat.tex}

\input{reduction-to-evaluation.tex}

\input{types.tex}

\input{automation.tex}
\input{extension.tex}

\input{related-work.tex}

\input{conclusion.tex}

\paragraph{Acknowledgements.}
Research was partially completed while the second and third authors were visiting the Institute for Mathematical Sciences, National University of Singapore in 2016.

\clearpage

%

%
%
\bibliographystyle{ACM-Reference-Format}
\bibliography{references}

\ifsupp
  \appendix

\input{apx-sec4.tex}

\input{apx-ty-assignment.tex}

\input{apx-inference.tex}

\input{apx-extension.tex}
\else
\fi

\end{document}

%% file: macros.tex
\newcommand{\mTo}{\mathrel{\To_m}}
\newcommand{\cTo}{\mathrel{\To_c}}

\newcommand{\sorts}{\vdash}
\newcommand{\types}{\vdash}
\newcommand{\monotypes}{\vdash}

\newcommand{\atheory}{\mathit{Th}}

\newcommand{\dto}[3]{#1{:}\term{#2} \to \term{#3}}
\newcommand{\boolsort}{o}
\newcommand{\intsort}{\mathsf{int}}

\renewcommand{\implies}{\Rightarrow}
\renewcommand{\iff}{\Leftrightarrow}
\newcommand{\mng}[1]{\llbracket #1 \rrbracket}
\newcommand{\mmng}[1]{\mathcal{M}\llbracket #1 \rrbracket}
\newcommand{\cmng}[1]{\mathcal{C}\llbracket #1 \rrbracket}
\newcommand{\smng}[1]{\mathcal{S}\llbracket #1 \rrbracket}

\newcommand{\rmmng}[1]{\mathcal{M}\llparenthesis #1 \rrparenthesis}
\newcommand{\rsmng}[1]{\mathcal{S}\llparenthesis #1 \rrparenthesis}

\newcommand{\subtype}{\sqsubseteq}
\newcommand{\order}{\mathsf{order}}
\newcommand{\abra}[1]{\langle #1 \rangle}

\newcommand{\ind}{\iota}
\newcommand{\truetm}{\mathsf{true}}
\newcommand{\falsetm}{\mathsf{false}}
\newcommand{\To}{\Rightarrow}

\newcommand{\andfn}{\mathsf{and}}

\newcommand{\orfn}{\mathsf{or}}

\newcommand{\notfn}{\mathsf{not}}
\newcommand{\existsfn}{\mathsf{exists}}
\newcommand{\mexistsfn}{\mathsf{mexists}}
\newcommand{\impliesfn}{\mathsf{implies}}
\newcommand{\forallfn}{\mathsf{forall}}
\newcommand{\aformulas}{\mathit{Fm}}
\newcommand{\aterms}{\mathit{Tm}}

\newcommand{\dom}{\mathsf{dom}}

\newcommand{\ifterm}[3]{\term{\mathsf{if} #1 \mathsf{then} #2 \mathsf{else} #3}}

\newcommand{\abs}[2]{\lambda #1.\,\term{#2}}
\newcommand{\basety}[2]{#1\langle \term{#2} \rangle}
\newcommand{\boolty}[1]{\basety{\boolsort}{#1}}
\newcommand{\intty}{\intsort}

\newcommand{\freshrel}[2]{\mathsf{freshRel}(#1)(#2)}
\newcommand{\freshty}[2]{\mathsf{freshTy}(#1)(#2)}
\newcommand{\freshenv}[1]{\mathsf{freshEnv(#1)}}

\newcommand{\vv}[1]{\overline{#1}}

\newcommand{\infers}{\Vdash}

\newcommand{\pideal}{{\Downarrow}}
\newcommand{\mm}{\medmid\medmid}
\newcommand{\FormatLabel}[1]{%
  \textsf{(#1)}
}
\newcommand{\RuleName}[2]{%
  \newcommand{#1}{\FormatLabel{#2}}
}
\newcommand{\makeset}[1]{\{ #1 \}}

\newcommand{\LSym}{\mathsf{LSym}}

\renewcommand{\phi}{\varphi}

\newcommand{\sfunc}{T^\mathcal{S}}
\newcommand{\mfunc}{T^\mathcal{M}}
\newcommand{\com}{\mathsf{Com}}
\newcommand{\lar}{\mathsf{Lar}}
\newcommand{\toptype}[1]{\top_{\!\!#1}}
\newcommand{\bottype}[1]{\bot_{#1}}

\RuleName{\ERefl}{ERefl}
\RuleName{\EDelta}{E$\delta$}
\RuleName{\EBeta}{E$\beta$}
\RuleName{\EAppL}{EAppL}
\RuleName{\EAppR}{EAppR}
\RuleName{\EAbs}{EAbs}
\RuleName{\EExTrue}{EExTrue}
\RuleName{\EExFalse}{EExFalse}
\RuleName{\EBinOp}{EBinOp}
\RuleName{\EUnOp}{EUnOp}

\RuleName{\SCst}{SCst}
\RuleName{\SVar}{SVar}
\RuleName{\SApp}{SApp}
\RuleName{\SAbs}{SAbs}

\RuleName{\TInt}{SubInt}
\RuleName{\TBool}{SubBool}
\RuleName{\TArrow}{SubArr}
\RuleName{\TProduct}{SubProd}
\RuleName{\TVar}{TVar}
\RuleName{\TConst}{TConst}
\RuleName{\TGen}{TGen}
\RuleName{\TConstraint}{TConstraint}
\RuleName{\TAnd}{TAnd}
\RuleName{\TOr}{TOr}
\RuleName{\TExists}{TExists}
\RuleName{\TSub}{TSub}
\RuleName{\TEquiv}{TEquiv}
\RuleName{\TStrengthen}{TStrengthen}
\RuleName{\TAppR}{TAppR}
\RuleName{\TAppI}{TAppI}
\RuleName{\TAbsR}{TAbsR}
\RuleName{\TAbsI}{TAbsI}
\RuleName{\TCond}{TCond}
\RuleName{\TLet}{TLet}
\RuleName{\TProg}{TProg}

\RuleName{\RInt}{RInt}
\RuleName{\RBool}{RBool}
\RuleName{\RArrow}{RArrow}
\RuleName{\RProd}{RProd}
\RuleName{\RScheme}{RScheme}
\RuleName{\REnvA}{REnv$_1$}
\RuleName{\REnvB}{REnv$_2$}
\RuleName{\REnvC}{REnv$_3$}

\RuleName{\ISubBool}{ISubBool}
\RuleName{\ISubInt}{ISubInt}
\RuleName{\ISubArrow}{ISubArrow}
\RuleName{\ISubProd}{ISubProd}
\RuleName{\IArrow}{IArrow}
\RuleName{\IAbsI}{IAbsI}
\RuleName{\IAbsR}{IAbsR}
\RuleName{\IVar}{IVar}
\RuleName{\IConst}{IConstraint}
\RuleName{\IPromote}{IPromote}
\RuleName{\IAppI}{IAppI}
\RuleName{\IAppR}{IAppR}
\RuleName{\IRedex}{IRedex}
\RuleName{\IAbs}{IAbs}
\RuleName{\IAnd}{IAnd}
\RuleName{\IOr}{IOr}
\RuleName{\IExists}{IExists}
\RuleName{\ILet}{ILet}
\RuleName{\IProg}{IProg}

\RuleName{\DBase}{TBase'}
\RuleName{\DArrow}{TArrow'}
\RuleName{\DVar}{TVar'}
\RuleName{\DConst}{TConst'}
\RuleName{\DPromote}{TPromote'}
\RuleName{\DApp}{TApp'}
\RuleName{\DRedex}{TRedex'}
\RuleName{\DAbs}{TAbs'}
\RuleName{\DCond}{TCond'}
\RuleName{\DLet}{TLet'}
\RuleName{\DSub}{TSub'}
\RuleName{\DWeaken}{TEquiv'}

\RuleName{\GVar}{GVar}
\RuleName{\GCst}{GCst}
\RuleName{\GConstraint}{GConstr}
\RuleName{\GAppRel}{GAppR}
\RuleName{\GAppInd}{GAppI}
\RuleName{\GAbs}{GAbs}
\RuleName{\GEx}{GEx}

\makeatletter
\def\term#1{%
    \@term#1 \@empty
}
\def\@term#1 #2{%
   \mathit{#1}
   \ifx #2\@empty\else
    \:\expandafter\@term  %
   \fi
   #2%
}
\makeatother

\usepackage{pict2e}

\makeatletter
\DeclareRobustCommand{\medmid}{%
  \mathrel{\mathpalette\med@mid\relax}%
}
\newcommand{\med@mid}[2]{%
  \sbox\z@{$\m@th#1\vdash$}%
  \setlength{\unitlength}{\ht\z@}%
  \begin{picture}(.3,1)
  \roundcap
  \linethickness{%
    \ifdim\unitlength<0.9ex
      0.08%
    \else
    \ifdim\unitlength<1.2ex
      0.07%
    \else
      0.06%
    \fi
  \fi\unitlength}
  \polyline(.2,0.03)(.2,0.97)
  \end{picture}%
}

\renewcommand{\models}{\vDash}

\newcommand{\Iter}{\mathit{Iter}}
\newcommand{\iter}{\mathit{iter}}

\newcommand{\add}{\mathit{add}}
\newcommand{\Add}{\mathit{Add}}

\renewcommand{\sum}{\mathit{sum}}

\newcommand{\lomon}{\mathsf{L}}
\newcommand{\upmon}{\mathsf{U}}
\newcommand{\iembed}{\mathsf{I}}
\newcommand{\jembed}{\mathsf{J}}

\newtheoremstyle{TheoremNum}
        {\topsep}{\topsep}              %
        {\itshape}                      %
        {}                              %
        {\bfseries}                     %
        {.}                             %
        { }                             %
        {\thmname{#1}\thmnote{ \bfseries #3}}%
    \theoremstyle{TheoremNum}
    \newtheorem{thmn}{Theorem}
    \newtheorem{lemn}{Lemma}

%% file: abstract.tex
Motivated by applications in automated verification of higher-order functional programs, we develop a notion of constrained Horn clauses in higher-order logic and a decision problem concerning their satisfiability.
We show that, although satisfiable systems of higher-order clauses do not generally have least models, there is a notion of canonical model obtained through a reduction to a problem concerning a kind of monotone logic program.
Following work in higher-order program verification, we develop a refinement type system in order to reason about and automate the search for models.
This provides a sound but incomplete method for solving the decision problem.
Finally, we show that there is a sense in which we can use refinement types to express properties of terms whilst staying within the higher-order constrained Horn clause framework.

%% file: introduction.tex
\section{Introduction}\label{sec:intro}

There is evidence to suggest that many first-order program verification problems can be framed as solvability problems for systems of constrained Horn clauses \cite{beyene-et-al-cav2013,bjorner-et-al-sas2013,bjorner-et-al-flc2015}, which are Horn clauses of first-order logic containing constraints expressed in some suitable background theory.
This makes the study of these systems particularly worthwhile since they provide a purely logical basis on which to develop techniques for first-order program verification.
For example, results on the development of highly efficient constrained Horn clause solvers  \cite{grebenschchikov-et-al-TACAS2012,Hoder-et-al-cav2011,gurfinkel-et-al-cav2015} can be exploited by a large number of program verification tools, each of which offloads some complex task (invariant finding is a typical example) to a solver by framing it in terms of constrained Horn clauses.

This paper concerns automated verification of higher-order, functional programs.
Whilst there are approaches to the verification of functional programs in which constrained Horn clause solving plays an important role, there is inevitably a mismatch between the higher-order nature of the program and the first-order logic in which the Horn clauses are expressed, and this must be addressed in some intelligent way by the program verifier.
For example, in recent work on refinement types \cite{rondon-et-al-pldi2008,vazou-et-al-icfp2015,unno-et-al-popl2013}, a type system is used to reduce the problem of finding an invariant for the higher-order program, to finding a number of first-order invariants of the ground-type data at certain program points.
This latter problem can often be expressed as a system of constrained Horn clauses.
When that system is solvable, the first-order invariants obtained can be composed in the type system to yield a higher-order invariant for the program (expressed as a type assignment).

In this paper we introduce \emph{higher-order constrained Horn clauses}, a natural extension of the notion of constrained Horn clause to higher-order logic, and examine certain aspects that we believe are especially relevant to applications in higher-order program verification, namely:
the existence of canonical solutions, the applicability of existing techniques to automated solving and the expressibility of program properties of higher type.

Let us elaborate on these goals and illustrate our motivation more concretely by discussing a particular example.
Consider the following higher-order program:
\[
  \begin{array}{l}
    \term{\mathsf{let} \add{} x y = x + y} \\
    \term{\mathsf{let} \mathsf{rec} \iter{} f s n} =
    \ifterm{n\leq0}{s}{f n (\iter{} f s (n - 1))} \\
    \mathsf{in}\;\term{\abs{n}{\mathsf{assert} (n \leq \term{\iter{} \add{} 0 n)}}}
  \end{array}
\]

\noindent
The term $\term{\iter{} \add{} 0 n}$ occurring in the last line of the program is the sum of the integers from $1$ to $n$ in case $n$ is non-negative and is $0$ otherwise.
Let us say that the program is \emph{safe} just in case the assertion is never violated, i.e. the summation is not smaller than $n$.

To verify safety, we must find an invariant that implies the required property.
For our purposes, an invariant will be an over-approximation of the input-output graphs of the functions defined in the program.
If we can find an over-approximation of the graph of the function $\term{\iter{} \add{} 0}$ which does not contain any pair $(n,m)$ with $n > m$, then we can be certain that the guard on the assertion is never violated.
Hence, we seek a set of pairs of natural numbers relating $n$ to $m$ at least whenever $\iter\ \add\ 0\ n$ evaluates to $m$ and which has no intersection with $>$.

The idea is to express the problem of finding such a program invariant logically, as a satisfiability problem for the following set of higher-order constrained Horn clauses:
\[
\begin{array}{l}
  \forall xyz.\, z = x + y \implies \term{\Add{} x y z} \\
  \forall fsnm.\, n \leq 0 \wedge m = s \implies \term{\Iter{} f s n m} \\
  \forall fsnm.\,n > 0 \wedge \big(\exists p.\,\term{\Iter{} f s (n-1) p} \wedge \term{f n p m}\big) \implies \term{\Iter{} f s n m} \\
  \forall nm.\, \term{\Iter{} \Add{} 0 n m} \implies n \leq m
\end{array}
\]
The clauses constrain the variables $\Add: \intsort \to \intsort \to \intsort \to \boolsort$ and $\Iter\!\!:(\intsort \to \intsort \to \intsort \to \boolsort) \to \intsort \to \intsort \to \intsort \to \boolsort$ with respect to the theory of integer linear arithmetic, so a model is just an assignment of particular relations\footnote{Throughout this paper we will speak of relations but work with their characteristic functions, which are propositional (Boolean-valued) functions, using $\boolsort$ for the sort of propositions.} to these variables that satisfies the formulas.
The first clause constrains $\Add$ to be an over-approximation of the graph of the addition function $\add{}$: whenever $z = x + y$, we at least know that $x,y$ and $z$ are related by $\Add$.
The second and third constrain $\Iter$ to be an over-approximation of the graph of the iteration combinator $\iter$.
Observe that the two branches of the conditional in the program appear as two clauses, expressing over-approximations in which the third input $n$ is at most or greater than $0$ respectively.
The final clause ensures that, taken together, these over-approximations are yet precise enough that they exclude the possibility of relating an input $n$ of $\iter\ \add\ 0$ to a smaller output $m$.
In other words, a model of these formulas \emph{is} an invariant, in the sense we have described above, and thus constitutes a witness to the safety of the program.

Notice that what we are describing in this example is a \emph{compositional approach}, in which an over-approximation $\Iter\ \Add\ 0:\intsort \to \intsort \to \boolsort$ to the graph of the function $\iter\ \add\ 0: \intsort \to \intsort$ is constructed from over-approximations $\Add$ and $\Iter$ of the graphs of the functions $\add$ and $\iter$.
Consequently, where $\iter$ was a higher-order function, $\Iter$ is a higher-order relation taking a ternary relation on integers as input, and the quantification $\forall f$ is over all such ternary relations $f$.
However, for the purposes of this paper, the details of how one obtains a system of higher-order constrained Horn clauses are not actually relevant, since we here study properties of such systems independently of how they arise.

\paragraph{Existence of canonical solutions.}
A set of higher-order constrained Horn clauses may have many models or none (consider that there are many invariants that can prove a safety property or none in case it is unprovable).
One model of the above set of clauses is the following assignment of relations (expressed in higher-order logic):
\[
    \Add{} \mapsto \abs{x \, y \, z}{z = x+y} \qquad
    \Iter{} \mapsto \abs{f\, s \, n \, m}{(\forall x y z.\,\term{f x y z} \implies 0 < x \implies y < z) \wedge{} 0 \leq s \implies n \leq m}
\]
Notice that this represents quite a large model (a coarse invariant).
For example, under this assignment the relation described by $\Iter\ \Add\ (-1)$ relates every pair of integers $n$ and $m$.
In the case of first-order constrained Horn clauses over the theory of integer linear arithmetic, if a set of clauses has a model, then it has a least model\footnote{In general, one can say that for each model of the background theory, a satisfiable set of clauses has a least satisfying valuation.}, and this least model property is at the heart of many of the applications of Horn clauses in practice.
If we consider the use of constrained Horn clauses in verification, a key component of the design of many successful algorithms for solving systems of clauses (and program invariant finding more generally) is the notion of approximation or abstraction.
However, to speak of approximation presupposes there is something to approximate.
For program verifiers there is, for example, the set of reachable states or the set of traces of the program, and for first-order constrained Horn clause solvers there is the least model.

In contrast, in Section \ref{sec:red-to-eval} we show that satisfiable systems of higher-order constrained Horn clauses do not necessarily have least models.
The problem, which has also been observed for pure (without constraint theory) higher-order Horn clauses by \citet{CharalambidisHRW13}, can be attributed to the use of unrestricted quantification over relations.
By restricting the semantics, so that interpretations range only over monotone relations (monotone propositional functions), we ensure that systems do have least solutions, but at the cost of %
\changed[lo]{abandoning the standard (semantics of) higher-order logic.}
The monotone semantics is natural but, for the purpose of specifying constraint systems, it can be unintuitive.
For example, consider the formula $\forall x.\,(\exists yz.\, x\ y \wedge y\ z) \implies P\ x$ which constrains $P : ((\intsort \to o) \to o) \to o$ so that it is at least true of all non-empty sets of non-empty sets of integers \footnote{Viewing relations of sort $\intsort \to o$ as sets of integers.}.
In the monotone semantics, this formula is guaranteed to have a least model, but inside that model $P$ is not true of the set $\{\{0\}\}$.

Ideally, we would like to be able to \emph{specify} constraint systems using the standard semantics of  higher-order logic, but \emph{solve} (build solvers for) systems in the monotone semantics.
In fact, we show that this is possible: we construct a pair of adjoint mappings with which the solutions to the former can be mapped to solutions of the latter and vice versa.
This allows us to reduce the problem of solving an system of constraints in the standard semantics to the problem of solving a system in the monotone semantics.
Monotonicity and the fact that satisfiable monotone systems have canonical solutions are key to the rest of the work in the paper.

\paragraph{Applicability of existing techniques to automated solving.}
Many of the techniques developed originally for the automation of first-order program verification transfer well to first-order constrained Horn clause solving.
Hence, to construct automated solvers for systems of  higher-order clauses, we look to existing work on higher-order program verification.
In automated verification for functional programs, one of the most successful techniques of recent years has been based on refinement type inference \cite{rondon-et-al-pldi2008,kobayashi-et-al-pldi2011,vazou-et-al-icfp2015,zhu-jagannathan-vmcai2013}.
The power of the approach comes from its ability to lift rich first-order theories over data to higher types using subtyping and the dependent product.

In Section \ref{sec:types}, we develop a refinement type system for higher-order constrained Horn clauses, in which types are assigned to the free relation variables that are being solved for.  The idea is that a valid type assignment is a syntactic representation of a model.  For example, the model discussed previously can be represented by the type assignment $\Gamma_I$:
\[
  \begin{array}{l}
  \Add\!\!: \dto{x}{\intty}{\dto{y}{\intty}{\dto{z}{\intty}{\boolty{z = x + y}}}} \\
  \Iter\!\!: (\dto{x}{\intty}{\dto{y}{\intty}{\dto{z}{\intty}{\boolty{0 < x \implies y < z}}}}) \to \dto{s}{\intty}{\dto{n}{\intty}{\dto{m}{\intty}{\boolty{0\leq s \implies n \leq m}}}}
  \end{array}
\]
The correspondence hinges on the definition of refinements $\boolty{\phi}$ of the propositional sort $\boolsort$, which are parametrised by a first-order constraint formula $\phi$ describing an upper bound on the truth of any inhabitant.
The dependent product and integer types are interpreted standardly, so that the first type above can be read as the set of all ternary relations on integers $x,y$ and $z$ that are false whenever $z$ is not $x + y$.

The system is designed so that its soundness allows one to conclude that a given first-order constraint formula can be used to approximate a given higher-order formula\footnote{Technically, the subjects of the type system are not all higher-order formulas but only the so-called \emph{goal} formulas.}.
Given a formula $G$, from the derivability of the judgement $\Gamma \types G : \boolty{\phi}$ it follows that $G \implies \phi$ in those interpretations of the relational variables that satisfy $\Gamma$.
For example, the judgement 
\[
\Gamma_I,\,n\!\!:\intsort,\,m\!\!:\intsort \types \term{\Iter{} \Add{} 0 n m} : \boolty{n \leq m}
\] 
is derivable, from which we may conclude that $n \leq m$ is a sound abstraction of $\Iter\ \Add\ 0\ n\ m$ in any interpretation of $\Iter$ and $\Add$ that  satisfies $\Gamma_I$.
This is a powerful assertion for automated reasoning because the formula $\phi$ in refinement type $\boolty{\phi}$ is a simple first-order constraint formula (typically belonging to a decidable theory) whereas the formula $G$ in the subject is a complicated higher-order formula, possibly containing relational variables whose meanings are a function of the whole system.
By adapting machinery developed for refinement type inference of functional programs, we obtain a sound (but incomplete) procedure for solving systems of higher-order constrained Horn clauses.
An implementation shows the method to be feasible.

\paragraph{Expressibility of program properties of higher type.}
We say that a property is of higher type if it is a property of a higher-order function.
It is possible to do whole-program verification in a higher-order setting using only properties of first-order type, because a complete program typically has a first-order type like $\intsort \to \intsort$.
However, it is also natural to want to specify properties of higher types, for example properties of higher-order functions of type $(\intsort \to \intsort) \to \intsort$.
Even when the ultimate goal is one of whole-program verification, being able to verify properties of higher types is advantageous because it can allow a large analysis to be broken down into smaller components according to the structure of the program.

However, the kinds of higher-type properties expressible by higher-order constrained Horn clauses is not immediately clear.
Therefore, we conclude Section \ref{sec:types} by showing that it is possible to state at least those properties that can be defined using refinement types, since their complements are expressible using goal terms.
\medskip

The rest of the paper is structured as follows.
In Section \ref{sec:prelims} we fix our presentation of higher-order logic and the notion of higher-order constrained Horn clause is made precise in Section \ref{sec:hornsat} along with the associated definition of solvability.
Section \ref{sec:red-to-eval} introduces monotone logic programs, which are better suited to automated reasoning, and shows that solvability of higher-order constrained Horn clause problems can be reduced to solvability of these programs.
This class of logic programs forms the basis for the refinement type system defined in Section \ref{sec:types}, which yields a sound but incomplete method for showing solvability through type inference.
An implementation of the method is also discussed in this section, which concludes by discussing the definability of higher-type properties defined by refinement types.
Finally, in Section \ref{sec:related-work}, we discuss related work and draw conclusions in Section \ref{sec:conclusion}.
\ifsupp
  Full proofs are included in the appendices.
\else
  Full proofs are included in the anonymous supplementary materials.
\fi

%% file: preliminaries.tex
\section{Higher-order logic}\label{sec:prelims}
We will work in a presentation of higher-order logic as a typed lambda calculus.

\paragraph{Sorts.}
Given a sort $\iota$ of individuals (for example $\intsort$), and a sort $\boolsort$ of propositions,
the general \emph{sorts} are just the simple types that can be built using the arrow: $\sigma \Coloneqq \iota \mid \boolsort \mid \sigma_1 \to \sigma_2$.
The \emph{order} of a sort $\sigma$, written $\order(\sigma)$, is defined as follows:
\[
  \order(\iota) = 1 \quad \order(\boolsort) = 1 \qquad \order(\sigma_1 \to \sigma_2) = \mathsf{max}(\order(\sigma_1) + 1, \order(\sigma_2))
\]
Note: we follow the convention from logic of regarding base sorts to be of order $1$.
Consequently, we will consider, for example, an $\intsort$ variable to be of order $1$ and an $(\intsort \to \boolsort) \to \intsort$ function to be of order $3$.

\paragraph{Terms.}
The terms that we consider are just terms of an applied lambda calculus.
We will write variables generally using $x,y,z$, or $X,Y,Z$ when we want to emphasise that they are of higher-order sorts.
\[
  \begin{array}{rcl}
    M,\,N & \Coloneqq & x \mid c \mid \term{M N} \mid \abs{x\!\!:\!\!\sigma}{M}
  \end{array}
\]
in which $c$ is a constant.
We assume that application associates to the left and the scope of the abstraction extends as far to the right as possible.
We identify terms up to $\alpha$-equivalence.

\paragraph{Sorting}
A \emph{sort environment}, typically $\Delta$,
is a finite sequence of pairs $x:\sigma$, all of whose subjects are required to be distinct.
We assume, for each constant $c$, a given sort assignment $\sigma_c$.
Then sorting rules for terms are, as standard, associated with the judgement $\Delta \types s : \sigma$ defined by:

\noindent
\vspace{2pt}
\begin{center}
\begin{minipage}{.27\linewidth}
\begin{prooftree}
\AxiomC{}
\LeftLabel{\SCst}
\UnaryInfC{$\Delta \sorts c : \sigma_c$}
\end{prooftree}
\end{minipage}
\begin{minipage}{.34\linewidth}
\begin{prooftree}
\AxiomC{}
\LeftLabel{\SVar}
\UnaryInfC{$\Delta_1, x:\sigma, \Delta_2 \sorts x : \sigma$}
\end{prooftree}
\end{minipage}\\\vspace{5pt}
\begin{minipage}{.45\linewidth}
\begin{prooftree}
\AxiomC{$\Delta \sorts s : \sigma_1 \to \sigma_2$}
\AxiomC{$\Delta \sorts t : \sigma_1$}
\LeftLabel{\SApp}
\BinaryInfC{$\Delta \sorts \term{s t} : \sigma_2$}
\end{prooftree}
\end{minipage}
\begin{minipage}{.45\linewidth}
\begin{prooftree}
\AxiomC{$\Delta, x:\sigma_1 \sorts s : \sigma_2$}
\LeftLabel{\SAbs}
\RightLabel{\makebox[0pt]{\hspace{50pt}$x \notin \dom(\Delta)$}}
\UnaryInfC{$\Delta \sorts \abs{x}{s} : \sigma_1 \to \sigma_2$}
\end{prooftree}
\end{minipage}
\end{center}
\vspace{8pt}

Given a sorted term $\Delta \types M : \sigma$, we say that a variable occurrence $x$ in $M$ is of \emph{order k}
just if the unique subderivation $\Delta' \types x : \sigma'$ rooted at this occurrence has $\sigma'$ of order $k$.
We say that a sorted term $\Delta \types M : \sigma$ is of \emph{order $k$} just if $k$ is the largest order of any of the variables occurring in $M$.

\paragraph{Formulas.}
Given a first-order signature $\Sigma$ specifying a collection of base sorts and sorted constants, we can consider higher-type \emph{formulas} over $\Sigma$, by considering terms whose constant symbols are either drawn from the signature $\Sigma$ or are a member of the following set \changed[lo]{$\LSym$} of logical constant symbols:
\[
  \begin{array}{ccc}
  \begin{array}{rcl}
  \truetm,\falsetm & : & \boolsort \\
  {\wedge},{\vee},{\implies} & : & \boolsort \to \boolsort \to \boolsort \\
  \end{array}
  &\phantom{woo}&
  \begin{array}{rcl}
  {\neg} &  : & \boolsort \to \boolsort \\
  {\forall_\sigma},{\exists_\sigma} & : & (\sigma \to \boolsort) \to \boolsort
  \end{array}
  \end{array}
\]
As usual, we write $\exists_\sigma(\abs{x\!\!:\!\!\sigma}{M})$ more compactly as $\exists x\!\!:\!\!\sigma .\, M$
and define the set of %
\emph{formulas} to be just the well-sorted terms of sort $\boolsort$.
In the context of formulas, it is worthwhile to recognise the subset of \emph{relational sorts}, typically $\rho$, which have the sort $\boolsort$ in tail position
and whose higher-order subsorts are also relational.
Formally:
\[
    \rho \Coloneqq \boolsort \mid \iota \to \rho \mid \rho \to \rho
\]
Since formulas are just terms, the notion of order carries over without modification.

\paragraph{Interpretation.}
Let $A$ be a $\Sigma$-structure.  In particular, we assume that $A$ assigns a non-empty set $A_\ind$ to each of the base sorts $\ind \in B$ and to the sort $\boolsort$ is assigned the distinguished lattice $\mathbbm{2} = \{0 \leq 1\}$.
We define the \emph{full sort frame} over $A$ by induction on the sort:
\[
  \smng{\iota} \coloneqq A_\ind
  \qquad\qquad \smng{\boolsort} \coloneqq \mathbbm{2}
  \qquad\qquad \smng{\sigma_1 \to \sigma_2} \coloneqq \smng{\sigma_1} \To \smng{\sigma_2}
\]
where $X \To Y$ is the full set-theoretic function space between sets $X$ and $Y$.
The lattice $\mathbbm{2}$ supports the following functions:
\[
  \begin{array}{cc}
  \begin{array}{rcl}
    \orfn(b_1)(b_2) &=& \max \{b_1,b_2\} \\
    \andfn(b_1)(b_2) &=& \min \{b_1,b_2\} \\
    \existsfn_\sigma(f) &=& \max \{ f(v) \mid v \in \mng{\sigma} \}
  \end{array}
  &
  \begin{array}{rcl}
    \notfn(b) &=& 1 - b \\
    \impliesfn(b_1)(b_2) &=& \orfn(\notfn(b_1))(b_2) \\
    \forallfn_\sigma(f) &=& \notfn (\existsfn_\sigma(\notfn \circ f))
  \end{array}
  \end{array}
\]
We extend the order on $\mathbbm{2}$ to order the set $\smng{\rho}$ of all relations of a given sort $\rho$ pointwise, defining the order $\subseteq_\rho$ inductively on the structure of $\rho$:
\begin{itemize}
  \item For all $b_1,b_2 \in \smng{\boolsort}$: if $b_1 \leq b_2$ then $b_1 \subseteq_{\boolsort} b_2$
  \item For all $r_1,r_2 \in \smng{\iota \to \rho}$: if, for all $n \in \smng{\iota}$, $r_1(n) \subseteq_{\rho} r_2(n)$, then $r_1 \subseteq_{\iota \to \rho} r_2$.
  \item For all $r_1,r_2 \in \smng{\rho_1 \to \rho_2}$: if, for all $s \in \smng{\rho_1}$, $r_1(s) \subseteq_{\rho_2} r_2(s)$, then $r_1 \subseteq_{\rho_1 \to \rho_2} r_2$.
\end{itemize}
This ordering determines a complete lattice structure on each $\smng{\rho}$, we will denote the (pointwise) join and meet by $\bigcup_{\rho}$ and $\bigcap_{\rho}$ respectively.  To aid readability, we will typically omit subscripts.

We interpret a sort environment $\Delta$ by the indexed product:
$
  \smng{\Delta} \coloneqq \Pi x \in \dom(\Delta).\smng{\Delta(x)}
$,
that is, the set of all functions on $\dom(\Delta)$ that map $x$ to an element of $\smng{\Delta(x)}$;
these functions, typically $\alpha$, are called \emph{valuations}.
We similarly order $\smng{\Delta}$ pointwise, with $f_1 \subseteq_{\Delta} f_2$ just if, for all $x\!\!:\!\!\rho \in \Delta$, $f_1(x) \subseteq_{\rho} f_2(x)$; thus determining a complete lattice structure.

For the purpose of interpreting formulas, we extend the structure $A$ to interpret the symbols from $\LSym$ according to the functions given above.
The interpretation of a term $\Delta \sorts M : \sigma$ is a function $\smng{\Delta \sorts M : \sigma}$ (we leave $A$ implicit) that belongs to the set $\smng{\Delta} \To \smng{\sigma}$, and which is defined by the following equations.
\[
\begin{array}{rcl}
\smng{\Delta \sorts x : \sigma}(\alpha) &=& \alpha(x) \\
\smng{\Delta \sorts c : \sigma}(\alpha) &=& c^A \\
\smng{\Delta \sorts \term{M N} : \sigma_2}(\alpha) &=& \smng{\Delta \sorts M : \sigma_1 \to \sigma_2}(\alpha)\big(\smng{\Delta \sorts N : \sigma_1}(\alpha)\big) \\
\smng{\Delta \sorts \abs{x:\sigma_1}{M} : \sigma_1 \to \sigma_2}(\alpha) &=& \abs{v \in \smng{\sigma_1}}{\smng{\Delta, x:\sigma_1 \sorts M : \sigma_2}(\alpha[x \mapsto v])} \\
\end{array}
\]

\paragraph{Satisfaction}
For a $\Sigma$-structure $A$, a formula $\Delta \types M : \boolsort$ and a valuation $\alpha \in \smng{\Delta}$, we say that $\abra{A,\alpha}$ \emph{satisfies} $M$ and write $A,\alpha \models M$ just if $\smng{\Delta \types M : \boolsort}(\alpha) = 1$.
We define entailment $M \models N$ between two formulas $M$ and $N$ in terms of satisfaction as usual.

%% file: hornsat.tex
\section{Higher-order constrained Horn clauses}\label{sec:hornsat}

We introduce a notion of constrained Horn clauses in higher-order logic.

\paragraph{Constraint language.}
Assume a fixed, first-order language over a first-order signature $\Sigma$,
consisting of: distinguished subsets of first-order terms $\aterms$ and first-order formulas $(\phi \in)$ $\aformulas$, and a first-order theory $\atheory$ in which to interpret those formulas.
We refer to this first-order language as the \emph{constraint language}, and $\atheory$ as the \emph{background theory}.

\paragraph{Atoms and constraints.}
An \emph{atom} is an applicative formula of shape $\term{X M_1 \cdots M_k}$
in which $X$ is a relational variable and each $M_i$ is a term.
A \emph{constraint}, $\phi$, is just a formula from the constraint language.
For technical convenience, we assume that atoms do not contain any constants (including logical constants),
and constraints do not contain any relational variables.

\paragraph{Constrained Horn clauses.}
Fix a sorting $\Delta$ of relational variables. %
The \emph{constrained goal formulas} over $\Delta$, typically $G$, and the \emph{constrained definite formulas} over $\Delta$, typically $D$, are the subset of all formulas defined by induction:
\[
  \begin{array}{rcl}
    G & \Coloneqq & M \mid \phi \mid G \wedge G \mid G \vee G \mid \exists x\!\!:\!\!\sigma.\: G  \\
    D & \Coloneqq & \truetm \mid \forall x\!\!:\!\!\sigma.\: D \mid D \wedge D \mid G \implies \term{X \vv{x}} \\
  \end{array}
\]
in which $\sigma$ is either the sort of individuals $\iota$ or a relational sort $\rho$, $M$ is an atom\footnote{\changed[lo]{We do not require the head variable of $M$ to be in $\dom(\Delta)$.}}, $\phi$ a constraint and, in the last alternative, $X$ is required to be a relational symbol inside $\dom(\Delta)$ and $\vv{x} = x_1 \cdots x_n$ a sequence of pairwise distinct variables.
It will often be convenient to view a constrained definite formula equivalently as a conjunction of (constrained) \emph{definite clauses}, which are those definite formulas with shape: $\forall \overline x . \, G \implies X\,\overline x$.

\begin{remark}\label{rem:definitional-fragment}
Our class of constrained definite formulas resembles the \emph{definitional} fragment of \cite{Wadge91}, due to the restrictions on the shape of $\term{X \vv{x}}$ occurring in the head of definite clauses.
However, the formalism discussed in loc.~cit., which was intended as a programming language, also restricted the existential quantifiers that could occur inside goal formulas and did not consider any notion of underlying constraint language.
\end{remark}

\paragraph{Problem.}
A (higher-order) \emph{Constrained Horn Clause Problem} is given by a tuple $\abra{\Delta,D,G}$ in which:
\begin{itemize}
\item $\Delta$ is a sorting of relational variables. %
\item $\Delta \sorts D : \boolsort$ is a constrained definite formula over $\Delta$.
\item $\Delta \sorts G : \boolsort$ is a constrained goal formula over $\Delta$.
\end{itemize}
The problem is of \emph{order} $k$ if $k$ is the largest order of the bound variables that occur in $D$ or $G$.
We say that such a problem is \emph{solvable} just if, for all models $A$ of the background theory $\atheory$, %
there exists a valuation $\alpha$ of the variables in $\Delta$ such that $A,\alpha \models D$, and yet $A,\alpha \not\models G$.

\begin{remark}
The presentation of the problem follows some of the literature for the use of first-order Horn clauses in verification.
The system of higher-order constrained Horn clauses is partitioned into two, distinguishing the definite clauses as a single definite formula and presenting the \emph{negation} of non-definite clauses as a single goal formula, which is required to be \emph{refuted} by valuations.
This better reflects the distinction between the program and the property to be proven.
Furthermore, solvability is defined in a way that allows for incompleteness in the background theory to be used to express under-specification of programming language features (for example, because they are difficult to reason about precisely).
\end{remark}

\begin{example}\label{ex:iter-horns}
Let us place the motivating system of clauses from the introduction formally into the framework.
To that end, let us fix the quantifier free fragment of integer linear arithmetic ($\mathsf{ZLA}$) as the underlying constraint language.
The sorting $\Delta$ of relational variables (the unknowns to be solved for) are given by:
\[
  \begin{array}{rl}
  \Add{}\!\!:&\intsort \to \intsort \to \intsort \to \boolsort \\
  \Iter{}\!\!:&(\intsort \to \intsort \to \intsort \to \boolsort) \to \intsort \to \intsort \to \intsort \to \boolsort
  \end{array}
\]
The higher-order constrained definite formula $D$ consists of a conjunction of the following three constrained definite clauses:
\[
  \begin{array}{l}
    \forall x \, y \, z.\, z = x + y \implies \term{\Add{} x y z} \\
    \forall f \, s\, n \, m.\, n \leq 0 \wedge m = 0 \implies \term{\Iter{} f s n m} \\
    \forall f \, s\, n \, m.\, (\exists p.\, n > 0 \wedge \term{\Iter{} f s (n-1) p} \wedge \term{f n p m}) \implies \term{\Iter{} f s n m}
  \end{array}
\]
Finally, the clause $\forall nm.\, \term{\Iter{} \Add{} 0 n m} \implies n \leq m$ expressing the property of interest is negated to give goal $G=\exists n \, m.\, \term{\Iter{} \Add{} 0 n m} \wedge m < n$.
This problem is solvable.
Being a complete theory, $\mathsf{ZLA}$ has one model up to isomorphism and, with respect to this model, the valuation given in the introduction satisfies $D$ but refutes $G$.
\end{example}

\vspace{15pt}

We believe that the generalisation of constrained horn clauses to higher orders is very natural.
However, our principal motivation in its study is the possibility of obtaining interesting applications in  higher-order program verification (analogous to those in first-order program verification with first-order constrained Horn clauses).
We interpret higher-order program verification in its broadest sense, encompassing not just purely functional languages but, more generally, problems in which satisfaction of a property depends upon an analysis of higher-order control flow.
For example, an early application of first-order constrained Horn clauses in the very successful constraint logic programming (CLP) paradigm of \cite{jaffar-maher-jlp1994} was the analysis of circuit designs \cite{heintze-et-al-jar1992}, for which systems of clauses were felt to be particularly suitable since they give a succinct, declarative specification of the analyses.
The relative advantages of circuit design description using higher-order combinator libraries or specification languages based on higher-order programming, such as \cite{bjesse-et-al-icfp1998}, are well documented, and systems of higher-order of constrained Horn clauses would therefore be a natural setting in which to verify the properties of such designs.

%% file: reduction-to-evaluation.tex
\section{Monotone models}\label{sec:red-to-eval}

One of the attractive features of \emph{first-order} constrained Horn clauses is that, for any given choice of interpretation of the background theory, every definite formula (set of definite clauses) possesses a unique, least model.
Consequently, it follows that there is a solution to a first-order Horn clause problem $\abra{\Delta,D,G}$ iff for each model of the background theory, the least model of $D$ refutes $G$.
This reformulation of the problem is of great practical benefit because it allows for the design of algorithms that, at least conceptually, exploit the canonicity.
For example, at the heart of the design of many successful algorithms for first-order Horn clause solving (and program invariant finding more generally) is the notion of approximation or abstraction.
However, to speak of approximation presupposes there is something to approximate.
For program verifiers there is, for example, the set of reachable states or the set of traces of the program, and for first-order constrained Horn clause solvers there is the least model.

The fact that first-order definite formulas possess a least model is paid for by restrictions placed on the syntax.
By forbidding negative logical connectives in goal formulas, it can be guaranteed that the unknown relation symbols in any definite clause occur positively exactly once, and hence obtaining the consequences of a given formula is a monotone operation.
We have made the same syntactic restrictions in our definition of higher-order constrained Horn formulas, but we do not obtain the same outcome.

\begin{theorem}\label{thm:no-least-models}
Higher-order constrained definite formulas do not necessarily possess least models.
\end{theorem}
\begin{proof}
Consider the sorting $\Delta_{\mathit{one}}$ of relational variables $P\!\!:((\mathsf{one}\to\boolsort) \to \boolsort)\to\boolsort$ and $Q\!\::\mathsf{one}\to\boolsort$ and the definite formula $D_{\mathit{one}}$:
\[
  \Delta \sorts \forall x.\, \term{x Q} \implies \term{P x} : o
\]
over a finite constraint language consisting of the sort $\mathsf{one}$ of individuals and no functions, relations or constants of any kind.
The language is interpreted in the background theory axiomatised by the sentence $\forall xy. x = y$, so that all models consist of a single individual $\smng{\mathsf{one}} = \{\star\}$.
Let us use $\mathbf{0}$ to denote the mapping $\star \mapsto 0$ and $\mathbf{1}$ denote the mapping $\star \mapsto 1$, both of which together comprise the set $\smng{\mathsf{one} \to \boolsort}$;
and let us name the elements of $\smng{(\mathsf{one} \to \boolsort) \to \boolsort}$ as follows:
\[
  \begin{array}{ccccccc}
    \mathsf{a} \coloneqq \begin{array}{rcl}
                        \mathbf{0} &\mapsto& 0 \\
                        \mathbf{1} &\mapsto& 1 \\
                      \end{array}
    &&
    \mathsf{b} \coloneqq \begin{array}{rcl}
                        \mathbf{0} &\mapsto& 0 \\
                        \mathbf{1} &\mapsto& 0 \\
                      \end{array}
    &&
    \mathsf{c} \coloneqq \begin{array}{rcl}
                        \mathbf{0} &\mapsto& 1 \\
                        \mathbf{1} &\mapsto& 1 \\
                      \end{array}
    &&
    \mathsf{d} \coloneqq \begin{array}{rcl}
                        \mathbf{0} &\mapsto& 1 \\
                        \mathbf{1} &\mapsto& 0 \\
                      \end{array}
    \\
  \end{array}
\]
Then we can describe minimal models $\alpha_1$ and $\alpha_2$ by the following equations:
\[
  \begin{array}{ccc}
    \begin{array}{ccc}
      \multicolumn{3}{c}{\alpha_1(Q) = \mathbf{0}} \\
      \alpha_1(P)(\mathbf{a}) = 0 &&
      \alpha_1(P)(\mathbf{b}) = 0 \\
      \alpha_1(P)(\mathbf{c}) = 1 &&
      \alpha_1(P)(\mathbf{d}) = 1 \\
    \end{array}
    &\hspace{10pt}&
    \begin{array}{ccc}
      \multicolumn{3}{c}{\alpha_2(Q) = \mathbf{1}} \\
      \alpha_2(P)(\mathbf{a}) = 1 &&
      \alpha_2(P)(\mathbf{b}) = 0 \\
      \alpha_2(P)(\mathbf{c}) = 1 &&
      \alpha_2(P)(\mathbf{d}) = 0 \\
    \end{array}
  \end{array}
\]
It is easy to verify that there are no models smaller than these and yet they are unrelated, so there is no least model.
\end{proof}
\noindent
A similar observation has been made in the pure (without constraint theory) setting by \citet{CharalambidisHRW13}.

Some consideration of the proof of this theorem leads to the observation that, despite an embargo on negative logical connectives in goal formulas, it may still be the case that unknown relation symbols occur negatively in goal formulas (and hence may occur positively more than once in a definite clause).
For example, consider the definite clause from above, namely: $\forall x.\, x\,Q \implies P\,x$.
Whether or not $Q$ can be said to occur positively in the goal formula $x\,Q$ depends on the action of $x$.
If we consider the subterm $\smng{x\,Q}$ as a function of $x$ and $Q$, then it is monotone in $Q$ only when the function assigned to $x$ is itself monotone.
By contrast, if $x$ is assigned an antitone function, as is the case when $x$ takes on the value $\mathbf{d}$, then $\smng{x\,Q}$ will be antitone in $Q$; for example $\alpha_1(Q) \subseteq \alpha_2(Q)$ but $\smng{x\,Q}(\alpha_2[x \mapsto \mathbf{d}]) \subseteq \smng{x\,Q}(\alpha_1[x \mapsto \mathbf{d}])$.

\subsection{Logic programs}
As we show in the following section, by restricting to an interpretation in which every function is monotone (in the logical order) we can obtain a problem in which there is a notion of least solution\footnote{Recall that \emph{least} (respectively \emph{monotone}) here refers to smallest in (preservation of) the \emph{logical} order, i.e.~with respect to inclusion of relations.}.
However, in doing so it seems that we sacrifice some of the logical purity of our original problem: if the universe of our interpretation contains only monotone functions, it does not include the function $\mathsf{implies}$ and so it becomes unclear how to interpret definite formulas.  Consequently, it requires a new definition of what it means to be a model of a formula.
Rather, the version of the problem we obtain by restricting to a monotone interpretation is much more closely related to work on the extensional semantics of higher-order \emph{logic programs} e.g. \cite{Wadge91,CharalambidisHRW13}, which emphasises the role of Horn clauses as definitions of rules.
Hence, we present the monotone restriction in those terms.

\paragraph{Goal terms.}
The class of well-sorted \emph{goal terms} $\Delta \sorts G : \rho$ is given by the sorting judgements defined by the rules below,
in which $c$ is one of $\wedge$, $\vee$ or $\exists_\sigma$ and here, and throughout the rules, $\sigma$ is required to stand for either the sort of individuals $\iota$ or  otherwise some relational sort.
It is easily verified that the constrained goal formulas are a propositional sorted subset of the goal terms.
From now on we shall use $G$, $H$ and $K$ to stand for arbitrary goal \emph{terms} and disambiguate as necessary.

\begin{table}[h]
  \def\arraystretch{3}
  \begin{tabular}{ccc}
      \AxiomC{}
      \LeftLabel{\GCst}
      \RightLabel{$c \in \makeset{\wedge, \vee, \exists_\iota} \cup \makeset{\exists_\rho \mid \rho}$}
      \UnaryInfC{$\Delta \sorts c : \rho_c$}
      \DisplayProof
    &&
      \AxiomC{}
      \LeftLabel{\GVar}
      \UnaryInfC{$\Delta_1, x:\rho, \Delta_2 \sorts x : \rho$}
      \DisplayProof
    \\
      \AxiomC{\vphantom{$\Delta$}}
      \RightLabel{$\Delta \sorts \phi : o \in \aformulas$}
      \LeftLabel{\GConstraint}
      \UnaryInfC{$\Delta \sorts \phi : o$}
      \DisplayProof
    &&
      \AxiomC{$\Delta, x:\sigma \sorts G : \rho$}
      \LeftLabel{\GAbs}
      \RightLabel{$x \notin \dom(\Delta)$}
      \UnaryInfC{$\Delta \sorts \abs{x}{G} : \sigma \to \rho$}
      \DisplayProof
    \\
      \AxiomC{$\Delta \sorts G : \iota \to \rho$}
      \RightLabel{$\Delta \sorts N : \iota \in \aterms$}
      \LeftLabel{\GAppInd}
      \UnaryInfC{$\Delta \sorts \term{G N} : \rho$}
      \DisplayProof
    &&
      \AxiomC{$\Delta \sorts G : \rho_1 \to \rho_2$}
      \AxiomC{$\Delta \sorts H : \rho_1$}
      \LeftLabel{\GAppRel}
      \BinaryInfC{$\Delta \sorts \term{G H} : \rho_2$}
     \DisplayProof
  \end{tabular}
\end{table}

\paragraph{Logic programs.}
A higher-order, constrained \emph{logic program}, $P$, over a sort environment $\Delta = x_1\!\!:\rho_1,\ldots,x_m:\rho_m$ is just a finite system of (mutual) recursive definitions of shape:
\[
x_1\!\!:\rho_1 = G_1,\quad \ldots,\quad x_m\!\!:\rho_m = G_m
\]
Such a program is well sorted when, for each $1 \leq i \leq m$, $\Delta \types G_i : \rho_i$.  Since each $x_i$ is distinct, we will sometimes regard a program $P$ as a finite map from variables to terms, defined so that $P(x_i) = G_i$.  We will write $\types P : \Delta$ to abbreviate that $P$ is a well-sorted program over $\Delta$.

\paragraph{Standard interpretation.}
Logic programs can be interpreted in the standard semantics by interpreting the right-hand sides of the equations using the term semantics given in Section \ref{sec:prelims}.
The program $P$ itself then gives rise to the functional $\sfunc_{P:\Delta} : \smng{\Delta} \To \smng{\Delta}$, sometimes called the \emph{one-step consequence operator} in the literature on the semantics of logic programming,
which is defined by: $\sfunc_{P:\Delta}(\alpha)(x) = \smng{\Delta \sorts P(x) : \Delta(x)}(\alpha)$.

\paragraph{The logic program of a definite formula.}
Every definite formula $D$ gives rise to a logic program, which is obtained by collapsing clauses that share the same head $\term{X \vv{x}}$ by taking the disjunction of their bodies, and viewing the resulting expression as a recursive definition of $X$.
The formulation as logic program is more convenient in two ways.
First, it is a more natural object to which to assign a monotone interpretation since we have eliminated implication, which does not act monotonically in its first argument, in favour of definitional equality.
Second, looking ahead to Section \ref{sec:types}, the syntactic structure of logic programs allows for a more transparent definition of a type system.

To that end, fix a definite formula $\Delta \sorts D : o$.
We assume, without loss of generality\footnote{Observe that such a shape can always be obtained by applying standard logical equivalences.}, that $D$ has the shape:
\[
  \forall \vv{x_{r_1}}.\,G_1 \implies \term{X_{r_1} \vv{x_{r_1}}} \quad\wedge\quad\cdots\quad\wedge\quad \forall \vv{x_{r_\ell}}.\, G_\ell \implies \term{X_{r_\ell} \vv{x_{r_\ell}}}
\]
over a sort environment $\Delta = \{X_1\!\!:\rho_1,\ldots,X_k\!\!:\rho_k\}$, i.e. $\{1,\ldots, k\} = \{r_1,\ldots,r_\ell\}$.
We construct a program over $\Delta$, called the \emph{logic program of $D$} and denoted $P_D$, as follows:
\[
  X_1 = \abs{\vv{x_1}}{G_1'},\quad \ldots{},\quad X_k = \abs{\vv{x_k}}{G_k'}
\]
where $G_j' = \bigvee \{ G_i \mid r_i = j \}$.  Note that $\{ G_i \mid r_i = j \}$ are exactly the bodies of all the definite clauses in $D$ whose heads are $X_j$.
The fact that $\sorts P_D : \Delta$ follows immediately from the well-sortedness of $D$.

\begin{example}\label{ex:iter-horn-program}
The definite formula component of the Horn clause problem from Example \ref{ex:iter-horns} is transformed into the following logic program $P$:
\[
    \Add{} = \abs{x \, y \, z}{z = x + y}  \qquad
    \Iter{} = \abs{f \, s\, n \, m}{({n \leq 0} \wedge {m = s})} \vee (\exists p.\, 0 < n \wedge \term{\Iter{} f s (n-1) p} \wedge{} \term{f n p m})
\]
\end{example}

\paragraph{Characterisation.}
If we were to consider only the standard interpretation then the foregoing development of logic programs would have limited usefulness.
As is well known at first-order, the definite formula and the program derived from it essentially define the same class of objects.
\begin{lemma}\label{lem:models-pfps}
For definite formula $D$, the prefixed points of $\sfunc_{P_D}$ are exactly the models of $D$.
\end{lemma}

\noindent
In contrast to the first-order case, it follows that $\sfunc_{P_D}$ does not have a least (pre-)fixed point{\footnote{\changed[sr]{In this paper we use the term \emph{prefixed point} to refer to those $x$ for which $f(x) \leq x$}}.
Indeed, for reasons already outlined, this functional is not generally monotone.
However, it will play an important role in Section \ref{sec:embedding}.

\subsection{Monotone semantics}

\begin{figure*}
\[
  \begin{array}{rcl}
    \mmng{\Delta \sorts x : \rho}(\alpha) &=& \alpha(x) \\
    \mmng{\Delta \sorts \phi : \boolsort}(\alpha) &=& \smng{\Delta \sorts \phi : \boolsort}(\alpha) \\
    \mmng{\Delta \sorts \term{G H} : \rho_2}(\alpha) &=& \mmng{\Delta \sorts G : \rho_1 \to \rho_2}(\alpha)(\mmng{\Delta \sorts H : \sigma_1}(\alpha)) \\
    \mmng{\Delta \sorts \term{G N} : \rho}(\alpha) &=& \mmng{\Delta \sorts G : \iota \to \rho}(\alpha)(\smng{\Delta \sorts N : \iota}(\alpha)) \\
    \mmng{\Delta \sorts \abs{x:\sigma}{G} : \sigma}(\alpha) &=& %
    \abs{x' \in \mmng{\sigma}}{\mmng{\Delta, x:\sigma \sorts G : \sigma}(\alpha[x \mapsto x'])}\\
    \mmng{\Delta \sorts \wedge : \boolsort \to \boolsort \to \boolsort}(\alpha) &=& \andfn \\
    \mmng{\Delta \sorts \vee : \boolsort \to \boolsort \to \boolsort}(\alpha) &=& \orfn \\
    \mmng{\Delta \sorts \exists_\sigma : (\sigma \to \boolsort) \to \boolsort}(\alpha) &=& \mexistsfn_\sigma
  \end{array}
\]
\caption{Monotone semantics of goal terms.}\label{fig:mono-term-semantics}
\end{figure*}

The advantage of logic programs is that they have a natural, monotone interpretation.

\paragraph{Monotone sort frame.}
We start from the interpretation of the background theory $A$, regarding $A_\iota$ as a discrete poset.  We then define the \emph{monotone sort frame} over $A$ by induction:
\[
   \mmng{\iota} \coloneqq A_\ind
   \qquad\quad \mmng{\boolsort} \coloneqq \mathbbm{2}
   \qquad\quad \mmng{\sigma_1 \to \sigma_2} \coloneqq \mmng{\sigma_1} \mTo \mmng{\sigma_2}
\]
where $X \mTo Y$ is the monotone function space between posets $X$ and $Y$,
i.e. the set of all functions $f \in X \To Y$ that have the property that $x_1 \leq x_2$ implies $f(x_1) \leq f(x_2)$.
It is easy to verify that this function space is itself a poset with respect to the pointwise ordering.
Of course, in case $X$ is discrete poset $A_\iota$, this coincides with the full function space.
We extend the lattice structure of $\mathbbm{2}$ to all relations $\mmng{\rho}$, analogously to the case of the full function space (and we reuse the same notation since there will be no confusion);
and we similarly define $\mmng{\Delta} \coloneqq \Pi x \in \dom(\Delta).\,\mmng{\Delta(x)}$.

It is worth considering the implications of monotonicity in the special case of relations, i.e. propositional functions.
A relation $r$ is an element of $X_1 \mTo \cdots{} \mTo X_k \mTo \mathbbm{2}$ just if it is \emph{upward closed}: whenever $r$ is true of $x_1,\ldots,x_k$ ($x_i \in X_i$), and $x_1',\ldots,x_k'$ ($x_i' \in X_i$) has the property that $x_i \subseteq x_i'$, then $r$ must also be true of $x_1',\ldots,x_k'$.  In particular, when $r \in X \mTo \mathbbm{2}$, then $r$ can be thought of as an upward closed set of elements of $X$.

\paragraph{Monotone interpretation.}
The interpretation of goal terms is defined in Figure \ref{fig:mono-term-semantics}.
As for the standard interpretation, we assume a fixed interpretation $A$ of the background theory, which is left implicit in the notation.
Whilst the standard interpretation of the positive logical constants for conjunction and disjunction will suffice, the interpretation of existential quantification needs to be relativised to the monotone setting:
$
  \mexistsfn_\sigma(r) = \mathit{max} \{r(d) \mid d \in \mmng{\sigma}\}
$.
Since the implication function $\impliesfn$ is not monotone (in its first argument), definite formulas are not interpretable in a monotone frame.
However, it is possible to interpret logic programs.
To that end, we define the functional $\mfunc_{P:\Delta}$ on semantic environments by: $\mfunc_{P:\Delta}(\alpha)(x) = \mmng{\Delta \types P(x) : \Delta(x)}(\alpha)$.
In analogy with the Horn clause problem, we call a prefixed point of $\mfunc_{P:\Delta}$ a \emph{model} of the program $P$.
This construction preserves the logical order.

\begin{lemma}
$\mmng{\Delta \sorts G : \rho} \in \mmng{\Delta} \mTo \mmng{\rho}$ \;and\; $\mfunc_{P:\Delta} \in \mmng{\Delta} \mTo \mmng{\Delta}$.
\end{lemma}
\begin{proof}
Immediately follows from the fact that $\mexistsfn$, $\andfn$ and $\orfn$ are monotone and all the constructions are monotone combinations.
\end{proof}

\noindent
It follows from the Knaster-Tarski theorem that, unlike the functional arising from the standard interpretation, the monotone functional $\mfunc_{P:\Delta}$ has a least fixed point, which we shall write $\mu \mfunc_{P:\Delta}$.
Consequently, logic programs $\sorts P : \Delta$ have a canonical monotone interpretation, $\mng{\sorts P : \Delta}$, which we define as $\mu \mfunc_{P:\Delta}$.

\paragraph{Monotone problem.}
By analogy with the first-order case, we are led to the following monotone version of the higher-order constrained Horn clause problem.  A \emph{Monotone Logic Program Safety Problem} (more often just \emph{monotone problem}) is a triple $(\Delta,P,G)$ consisting of a sorting of relational variables $\Delta$, a logic program $\sorts P : \Delta$ and a goal $\Delta \sorts G : o$.  The problem is solvable just if, for all models of the background theory, there is a prefixed point $\alpha$ of $\mfunc_{P:\Delta}$ such that $\mmng{G}(\alpha) = 0$.

\subsection{Canonical embedding}\label{sec:embedding}
In the monotone problem we have obtained a notion of safety problem that admits a least solution.
Due to the monotonicity of $\mmng{G}$, there is a prefix point witnessing solvability iff the least prefix point is such a witness, i.e. iff $\mmng{G}(\mmng{P}) = 0$.
This clears the way for our algorithmic work in Section \ref{sec:types}, which consists of  apparatus in which to construct sound approximations of $\mmng{P}$.
However, the price we have had to pay seems severe, since we have all but abandoned our original problem definition.

The monotone logic program safety problem lacks the logical purity of the higher-order constrained Horn clause problem, which is stated crisply in terms of the standard interpretation of higher-order logic and the usual notion of models of formulas.
In the context of program verification, the monotone problem appears quite natural, but if we look further afield, to e.g. traditional applications of constrained Horn clauses in constraint satisfaction, it seems a little awkward.
For example, the significance of allowing only monotone solutions seems unclear if one is looking to state a scheduling problem for a haulage company or a packing problem for a factory.

Ideally, we would like to \emph{specify} constraint systems using the standard Horn clause problem, with its clean logical semantics, but \emph{solve} instances of the monotone problem, which is easier to analyse, due to monotonicity and the existence of canonical models.
In fact, we shall show that this is possible: every solution to the Horn clause problem $\abra{\Delta,D,G}$ determines a solution to the monotone problem $\abra{\Delta,P_D,G}$ and vice versa (Theorem~\ref{thm:reduction-to-mono}).

\paragraph{Transferring solutions}
Let us begin by considering what a mapping between solutions of $\abra{\Delta,P_D,G}$ and solutions of $\abra{\Delta,D,G}$ would look like.
In both cases, a solution is a model: the former is a mapping from variables to monotone relations and the latter is a mapping from (the same) variables to arbitrary relations.

At first glance, it might appear that one can transfer models of $P_D$ straightforwardly to models of $D$, because monotone relations are, in particular, relations.  However, the situation is a little more difficult.
Although the solution space of the Horn clause problem is larger, more is required of a valuation in order to qualify as a model because the constraints of the Horn clause problem,
which involve universal quantification over all relations,
are more difficult to satisfy than the equations of the monotone problem, which involve (implicitly) quantification over only the monotone relations.

To see this concretely, it is useful to consider the simpler case in which all relations are required to  be unary, i.e. $\rho$ is of shape $(\cdots((\iota \to \boolsort) \to \boolsort) \to \cdots{}) \to \boolsort$.
In the unary case, we can think of a relation simply as describing a set of objects, where those objects may themselves be sets of objects.
For example, $\smng{(\iota \to \boolsort) \to \boolsort}$ describes the collection of all sets of sets of individuals.
On the other hand, the constraint on monotonicity of relations has the consequence that, if we think of $\mmng{(\iota \to \boolsort) \to \boolsort}$ as describing a collection of sets, it is the collection only consisting of those sets of sets of individuals that are upward closed.
That is, a set $s$ is in $\mmng{(\iota \to \boolsort) \to \boolsort}$ just if, whenever a set of individuals $t$ is in $s$ and $t \subseteq u$ then $u$ is also in $s$.
In general, we can think of $\smng{\sigma \to \boolsort}$ as the collection of all sets of objects from $\smng{\sigma}$,
and $\mmng{\sigma \to \boolsort}$ as the collection of hereditarily upward-closed sets.
Now consider the logic program $P = \abs{x}{\truetm}$, in which $P$ is of sort $((\intsort \to \boolsort) \to \boolsort) \to \boolsort$.  One model of this program is to take for $P$ the set of \emph{all} upward-closed sets of sets of integers, which is a relation in $\mmng{((\intsort \to \boolsort) \to \boolsort) \to \boolsort}$.  However, the set of all upward-closed sets of sets of integers is not a model of the corresponding formula $\forall x.\, \truetm \implies P\ x$ in the standard semantics, because it does not contain, for example, the set $\{\{0\}\}$ which is not upward closed
(i.e.~its characteristic function is not a \emph{monotone} Boolean-valued function),
yet the universal quantification requires it.

So, although there is a canonical inclusion of $\mmng{\rho}$ into $\smng{\rho}$, it does not extend to map models of $P_D$ to models of $D$ in general.
If we return to thinking of the elements of $((\intsort \to \boolsort) \to \boolsort) \to \boolsort$ formally as Boolean-valued functions, the inclusion described above is mapping the monotone function $r \in \mmng{(\intsort \to \boolsort) \to \boolsort)} \mTo \mathbbm{2}$, which satisfies $r(t) = 1$ for all $t \in \mmng{(\intsort \to \boolsort) \to \boolsort}$, to the function $J(r) \in \smng{(\intsort \to \boolsort) \to \boolsort} \To \mathbbm{2}$, which satisfies, for all $t \in \smng{(\intsort \to \boolsort) \to \boolsort)}$:
\[
  J(r)(t) = \begin{cases} r(t) & \text{if $t \in \mmng{(\intsort \to \boolsort) \to \boolsort)}$} \\ 0 & \text{otherwise} \end{cases}
\]
In other words, it lifts a function whose domain consists only of hereditarily monotone relations to a function whose domain consists of all relations simply by mapping non-hereditarily monotone inputs to $0$.  We could equally well consider the dual, in which all such inputs were mapped to $1$, but the image of the mapping would typically not refute the goal $G$ because the models so constructed are too large.

This counterexample suggests that we require a mapping of monotone relations $r \in \mmng{((\intsort \to \boolsort) \to \boolsort) \to \boolsort}$ to standard relations $J(r) \in \smng{((\intsort \to \boolsort) \to \boolsort) \to \boolsort}$ that is a little more sophisticated in the action of $J(r)$ on inputs that are not hereditarily monotone.
Instead of mapping all such inputs to $0$ or all such inputs to $1$ we shall determine the value of $J(r)$ on some non-monotone input $t \in \smng{(\intsort \to \boolsort) \to \boolsort}$ by considering the value of $r$ on a monotone input $U(t) \in \mmng{(\intsort \to \boolsort) \to \boolsort}$ which is somehow \emph{close} to $t$.
In fact there are two canonical choices of hereditarily monotone relations close to a given relation $t$, which are obtained as, respectively, the largest monotone relation included in $t$ and the smallest monotone relation in which $t$ is included.
We will describe the situation in general using Galois connection.

\paragraph{Galois connection.}
A pair of functions $f:P \to Q$ and $g:Q \to P$ between partial orders $P$ and $Q$ is a \emph{Galois connection} just if, for all $x \in P$ and $y \in Q$: $f(x) \leq y$ iff $x \leq g(y)$.  In such a situation we write $f \dashv g$ and $f$ is said to be the \emph{left adjoint} of $g$, and $g$ the \emph{right adjoint} of $f$.
First, it is easy to verify that if $f \dashv g$ then $f$ and $g$ are monotone.

\begin{proposition}
Given a pair of monotone maps $f : P \to Q$ and $g : Q \to P$, the following are equivalent:
\begin{enumerate}[(1)]
\item The pair $(f, g)$ is a Galois connection.
\item $f \circ g \leq 1_Q$ and $g \circ f \leq 1_P$.
\item For all $x \in P$, $\inf \, \{y \in Q \mid x \leq g(y) \}$ is defined and equal to $f(x)$; and for all $y \in Q$,
$\sup \, \{x \in P \mid f(x) \leq y \}$ is defined and equal to $g(y)$.
\end{enumerate}
Further, if any one of the above conditions holds, then
\begin{enumerate}
\item[(4)] $f$ preserves all existing suprema, and $g$ preserves all existing infima.
\item[(5)] $f = f \circ g \circ f$ and $g = g \circ f \circ g$.
\end{enumerate}
\end{proposition}
To see the Proposition, just view the pair $(f, g)$ as functors on categories \cite{MacLane71};
then they forms a Galois connection exactly when they are an adjunction pair.
The following facts are easy to verify:
\begin{enumerate}[(i)]
  \item If $P$ is a complete lattice and $f : P \to Q$ preserves all joins, then $f$ is a left adjoint.
  \item If $Q$ is a complete lattice and $g : Q \to P$ preserves all meets, then $g$ is an right adjoint.
  \item If $f_1:P \to Q$, $g_1:Q \to P$, $f_2:Q \to R$ and $g_2:R \to Q$ with $f_1 \dashv g_1$ and $f_2 \dashv g_2$ then it follows that $f_1 \circ f_2 \dashv g_1 \circ g_2$,
  is a Galois connection between partial orders $P$ and $R$.
  \item If $f_1:P_1 \to Q_1$, $g_1:Q_1 \to P_1$, $f_2:P_2 \to Q_2$ and $g_2:Q_2 \to P_2$ with $f_1 \dashv g_1$ and $f_2 \dashv g_2$ then it follows that the pair of functions $f: [P_1 \mTo P_2] \to [Q_1 \mTo Q_2]$ and $g:[Q_1 \mTo Q_2] \to [P_1 \mTo P_2]$, defined by: $f(h) = f_1 \circ h \circ g_2$ and $g(k) = g_1 \circ k \circ f_2$
  is a Galois connection $f \dashv g$ between the corresponding monotone function spaces (ordered pointwise).
\end{enumerate}
Facts (i) and (ii) are just the Adjoint Functor Theorem (see e.g.~\cite{MacLane71}) specialised to the case of preorders.
Notice that a special case of (iv) is the construction of a Galois connection $f \dashv g$ between the (full) function spaces $A \To P$ and $A \To Q$ (with pointwise order) for fixed set $A$ and partial orders $P$ and $Q$, whenever there is a Galois connection $f_2 \dashv g_2$ between $P$ and $Q$.
This is because there is always a trivial Galois connection $\mathsf{id} \dashv \mathsf{id}$ on any set $A$ by viewing it as a discrete partial order.

\paragraph{Embedding the monotone relations.}
For general $\rho$, every complete lattice of monotone relations $\mmng{\rho}$ can be embedded in the complete lattice of all relations $\smng{\rho}$ in the following two ways.
\begin{equation}
  \smng{\rho} \galoiS{\lomon_\rho}{\iembed_\rho} \mmng{\rho} \Galois{\jembed_\rho}{\upmon_\rho} \smng{\rho}
  \label{eq:monoembed}
\end{equation}
We define the family of right adjoints $\iembed_\rho$ and the family of left adjoints $\jembed_\rho$, by induction on the sort $\rho$.  In the definition, $\lomon_{\rho}$ is the uniquely determined left adjoint of $\iembed_\rho$ and $\upmon_{\rho}$ is the uniquely determined right adjoint of $\jembed_\rho$.
\[
\begin{array}{ccc}
  \begin{array}{rcl}
    \iembed_{\boolsort}(b) &=& b \\
    \iembed_{\iota \to \rho}(r) &=& \iembed_\rho \circ r \\
    \iembed_{\rho_1 \to \rho_2}(r) &=& \iembed_{\rho_2} \circ r \circ \lomon_{\rho_1}
  \end{array}
&
\hspace{10pt}
&
  \begin{array}{rcl}
    \jembed_{\boolsort}(b) &=& b \\
    \jembed_{\iota \to \rho}(r) &=& \jembed_\rho \circ r \\
    \jembed_{\rho_1 \to \rho_2}(r) &=& \jembed_{\rho_2} \circ r \circ \upmon_{\rho_1}
  \end{array}
\end{array}
\]

We briefly discuss this definition before verifying its correctness.  It is worth observing that, rather than defining $\iembed_{\rho_1 \to \rho_2}$ and $\jembed_{\rho_1 \to \rho_2}$ using the induced left and right adjoints at $\rho_1$, we could have given the definition explicitly (recalling Galois connection properties (ii) and (iii)) by:
\[
  \iembed_{\rho_1 \to \rho_2}(r)(s) = \iembed_{\rho_2}\big(r(\bigcap\{t \mid s \subseteq \iembed_{\rho_1}(t) \})\big)
  \quad\text{and}\quad \jembed_{\rho_1 \to \rho_2}(r)(s) = \jembed_{\rho_2}\big(r(\bigcup\{t \mid \jembed_{\rho_1} \subseteq s\})\big).
\]
We have not given the definition in this way because the proofs that follow only require the adjunction properties of $\lomon_{\rho_1}$ and $\upmon_{\rho_1}$, and not any explicit characterisation.
To unpack the definition a little more, suppose $\rho$ is restricted to unary relations and consider the first few elements of this inductive family.
When $\rho$ is either $o$ or $\iota \to o$, $\smng{\rho} = \mmng{\rho}$, and $\iembed$ and $\jembed$ are both the identity.
Consequently, they are both left and right adjoint to themselves, so that $\lomon$ and $\upmon$ are also both the identity.
When $\rho$ is $(\iota \to \boolsort) \to \boolsort$, by definition $\jembed_\rho(r) = \jembed_o \circ r \circ \upmon_{\iota \to \boolsort}$ but, as discussed, both of $\jembed_o$ and $\upmon_{\iota \to \boolsort}$ are identities on their respective domains, so $\jembed_\rho(r)$ is just $r$.
However, $\mmng{(\iota \to \boolsort) \to \boolsort}$ is strictly contained within $\smng{(\iota \to \boolsort) \to \boolsort}$, so $\jembed_\rho$ is merely an inclusion and, consequently, the induced right adjoint $\upmon_\rho$ is more interesting.
Using Galois connection property (ii) it can be computed explicitly, revealing that it maps each $s \in \smng{\rho}$ to $\bigcup \{ t \in \mmng{\rho} \mid \jembed_{\rho}(t) \subseteq s \}$.
But, we have seen that $\jembed_{(\iota \to \boolsort) \to \boolsort}(t) = t$, so it follows that $\upmon_\rho(s)$ is just the largest monotone relation included in $s$, and we arrive back at the discussion with which we started this subsection.
The following proof gives more insight on the structure of the mappings.

\begin{lemma}\label{galois-relations}
  For each $\rho$, (i) $\lomon_\rho \dashv \iembed_\rho$ and (ii) $\jembed_\rho \dashv \upmon_\rho$ are well-defined Galois connections.
\end{lemma}
\begin{proof}
  We prove only (i) because the proof of (ii) is analogous.  We show that $\iembed_\rho : \mmng{\rho} \to \smng{\rho}$ is a well-defined right adjoint by induction on $\rho$.
  \begin{itemize}
  \item When $\rho$ is $o$, $\mmng{o} = \smng{o}$ and $\iembed_o$ is the identity, which has left adjoint also the identity.
  \item When $\rho$ is of shape $\iota \to \rho_2$, it follows from the induction hypothesis that $\iembed_{\rho_2} : \mmng{\rho_2} \to \smng{\rho_2}$ is a well-defined right adjoint.  It follows from Galois connection property (vi) that the mapping $r \mapsto \iembed_{\rho_2} \circ r : \mmng{\iota \to \rho_2} \to \smng{\iota \to \rho_2}$ is a well-defined right adjoint.
  \item Finally, when $\rho$ has shape $\rho_1 \to \rho_2$, we decompose the definition of $\iembed_\rho$ as follows:
  \[
    \mmng{\rho_1} \mTo \mmng{\rho_2} \xrightarrow{r \;\mapsto\; \iembed_{\rho_2} \:\circ\: r \:\circ\: \lomon_{\rho_1}} \smng{\rho_1} \mTo \smng{\rho_2} \xrightarrow{s \;\mapsto\; s} \smng{\rho_1} \To \smng{\rho_2}
  \]
  It follows from the induction hypothesis that $\iembed_{\rho_1} : \mmng{\rho_1} \to \smng{\rho_1}$ and $\iembed_{\rho_2} : \mmng{\rho_2} \to \smng{\rho_2}$ are both well-defined right-adjoints, from which we may infer the existence of left adjoint $\lomon_{\rho_1} : \smng{\rho_1} \to \mmng{\rho_1}$.  It follows from Galois connection property (iv) that the mapping $r \mapsto \iembed_{\rho_2} \circ r \circ \lomon_{\rho_1} : \mmng{\rho_1} \mTo \mmng{\rho_2} \to \smng{\rho_1} \mTo \smng{\rho_2}$ is a well-defined right adjoint (with codomain the monotone function space).  Finally, observe that there is a canonical inclusion between the monotone and full function spaces which, since it trivially preserves meets, is as an right adjoint according to Galois connection property (ii).
  The result follows since right adjoints compose (Galois connection property (iii)).
  \end{itemize}
\end{proof}

The Galois connections give a canonical way to move between the universes of monotone and arbitrary relations;
we extend them to mappings on valuations $\alpha \in \mmng{\Delta}$ by:
\[
  \iembed_\Delta(\alpha)(x) =
    \begin{cases}
      \alpha(x) & \text{if $\Delta(x) = \iota$} \\
      \iembed_{\Delta(x)}(\alpha(x)) & \text{otherwise} \\
    \end{cases}
  \qquad\qquad\qquad
  \jembed_\Delta(\alpha)(x) =
    \begin{cases}
      \alpha(x) & \text{if $\Delta(x) = \iota$} \\
      \jembed_{\Delta(x)}(\alpha(x)) & \text{otherwise} \\
    \end{cases}
\]
The action is pointwise on relations and trivial on individuals.  It is easy to verify that each $\iembed_\Delta$ and $\jembed_\Delta$ are right and left adjoints respectively.

\begin{corollary}\label{lem:galois-valuations}
For each sorting $\Delta$, (i) $\lomon_\Delta \dashv \iembed_\Delta$ and (ii) $\jembed_\Delta \dashv \upmon_\Delta$ are well-defined Galois connections.
\end{corollary}

\noindent
We now have a canonical way to move between monotone and arbitrary valuations, but our aim was to be able to map models of $P_D$ to models of $D$ (and vice versa), and we do not yet have any evidence that our mappings are at all useful in this respect.
In both cases, models are prefixed points of certain functionals so we look for conditions which ensure that mappings preserve the property of being a prefix point.
One such condition is the following: if $F : P \to Q$ is monotone, $T_1 : P \to P$ and $T_2: Q \to Q$ then $F$ will send prefixed points of $T_1$ to prefixed points of $T_2$ whenever it satisfies $T_2 \circ F \subseteq F \circ T_1$ (in the pointwise order).
This is because if $T_1(x) \leq x$ then $F(T_1(x)) \leq F(x)$ by monotonicity, but also $T_2(F(x)) \leq F(T_1(x))$ by the assumption so that $T_2(F(x)) \leq F(x)$.

In the following we will prove that the right adjoints $\iembed$ and $\upmon$ preserve prefix points by showing that $\mfunc \circ \upmon \subseteq \upmon \circ \sfunc$ and $\sfunc \circ \iembed \subseteq \iembed \circ \mfunc$.
The meat of the definitions of $\sfunc$ and $\mfunc$ lies in the semantics of goal terms, so we first show that, for all goal terms $G$, $\mmng{G} \circ \upmon \subseteq \upmon \circ \smng{G}$ and $\smng{G} \circ \iembed \subseteq \iembed \circ \mmng{G}$.
However, we rephrase $\mmng{G} \circ \upmon \subseteq \upmon \circ \smng{G}$ equivalently as $\jembed \circ \mmng{G} \circ \upmon \subseteq \smng{G}$ and $\smng{G} \circ \iembed \subseteq \iembed \circ \mmng{G}$ equivalently as $\smng{G} \subseteq \iembed \circ \mmng{G} \circ \lomon$, which allows for a straightforward induction.

\begin{lemma}\label{lem:embed-mngs}
For all goal terms $\Delta \sorts G : \rho$,\ \
$\jembed_{\rho} \circ \mmng{G} \circ \upmon_{\Delta} \; \subseteq \; \smng{G} \; \subseteq \; \iembed_{\rho} \circ \mmng{G} \circ \lomon_{\Delta}$.
\end{lemma}

\noindent
The fact that the two right adjoints map between prefix points now follows immediately once the equivalence of our rephrasings have been verified.

\begin{lemma}[Model translation]\label{lem:prefix-point-mapping}
Fix a program $\sorts P : \Delta$.
\begin{enumerate}[(i)]
  \item If $\beta$ is a prefixed point of $\sfunc_{P:\Delta}$ then $\upmon_\Delta(\beta)$ is a prefixed point of $\mfunc_{P:\Delta}$.
  \item If $\alpha$ is a prefixed point of $\mfunc_{P:\Delta}$, then $\iembed_\Delta(\alpha)$ is a prefixed point of $\sfunc_{P:\Delta}$.
\end{enumerate}
\end{lemma}
\begin{proof}
It follows from Lemma \ref{lem:embed-mngs} and Lemma \ref{lem:galois-valuations} that, for any goal term $G$:
\[
  \mmng{G}(\upmon(\beta)) \subseteq \upmon(\smng{G}(\beta))
  \quad\text{and}\quad
  \smng{G}(\iembed(\alpha)) \subseteq \iembed(\mmng{G}(\alpha))
\]
The first follows from $\jembed \circ \mmng{G} \circ \upmon \subseteq \smng{G}$ since $\jembed$ is left adjoint.
The second follows from $\smng{G} \subseteq \iembed \circ \mmng{G} \circ \lomon$ by pre-composing with $\iembed$ on both sides and noting that $\lomon \circ \iembed$ is deflationary.
By definition: $\mfunc(\alpha)(x) = \mmng{P(x)}(\alpha)$ and $\sfunc(\beta)(x) = \smng{P(x)}(\beta)$, so that we can deduce the following from the above inclusions:
\[
  \mfunc(\upmon(\beta)) \subseteq \upmon(\sfunc(\beta))
  \quad\text{and}\quad
  \sfunc(\iembed(\alpha)) \subseteq \iembed(\mfunc(\alpha))
\]
For part (i), it only remains to observe that if $\sfunc(\beta) \subseteq \beta$ then, by monotonicity, $\upmon(\sfunc(\beta)) \subseteq \upmon(\beta)$ and, by the above inclusion, $\mfunc(\upmon(\beta)) \subseteq \upmon(\beta)$. Part (ii) is analogous.
\end{proof}

\noindent
Moreover, these mappings preserve refutation of the goal.  Intuitively, we think of $\upmon(\beta)$ (respectively $\iembed(\alpha)$) as the largest monotone (respectively standard) valuation that is smaller than $\beta$ (respectively $\alpha$).
So, (as is made precise in Lemma \ref{lem:embed-mngs}) if the latter refutes a goal, so should the former.
Hence, we obtain the following problem reduction.

\begin{theorem}\label{thm:reduction-to-mono}
The higher-order constrained Horn clause problem $\abra{\Delta,D,G}$ is solvable, 
iff the monotone logic program safety problem $\abra{\Delta,P_D,G}$ is solvable, and 
iff in all models of the background theory $\mmng{G}(\mmng{P_D}) = 0$.
\end{theorem}
\begin{proof}
We prove a chain of implications.
\begin{itemize}
\item Assume that $\abra{\Delta,D,G}$ is solvable, so that for each model $A$ of the background theory, there is a valuation $\beta$ and $A,\,\beta \models D$ and $A,\,\beta \not\models G$, i.e. $\smng{G}(\beta) = 0$.
Fix such a model $A$ of the background theory and then let $\beta$ be the witness given above.
Then it follows from Lemma \ref{lem:embed-mngs} that $\jembed(\mmng{G}(\upmon(\beta))) \subseteq 0$, i.e. $\mmng{G}(\upmon(\beta)) = 0$.
Since $\beta$ is a model of $D$, it follows from Lemma \ref{lem:models-pfps} that it is also a prefixed point of $\sfunc_{P_D}$ and hence $\upmon(\beta)$ is a prefixed point of $\mfunc_{P_D}$ by Lemma \ref{lem:prefix-point-mapping}.  Therefore, $\abra{\Delta,P_D,G}$ is also solvable.
\item If $\abra{\Delta,P_D,G}$ is solvable, then, for each model $A$ of the background theory, there is a prefix point $\alpha$ of $\mfunc_{P_D}$ and $\mmng{G}(\alpha) = 0$.
However, $\mmng{P_D}$ is, by definition, the least prefixed point so $\mmng{P_D} \subseteq \upmon(\alpha)$ and hence $\mmng{G}(\mmng{P_D}) = 0$ follows by monotonicity.
\item Finally assume that in all models of the background theory $\mmng{G}(\mmng{P_D}) = 0$.
Fix such a model of the background theory.
We claim that $\lomon(\smng{G}(\iembed(\mmng{P_D}))) \subseteq 0$ follows from Lemma \ref{lem:embed-mngs}, which is to say that $\smng{G}(\iembed(\mmng{P_D}))) = 0$.
To see this, observe that $\smng{G}(\iembed(\mmng{P_D})) \subseteq \iembed(\mmng{G}(\mmng{P_D}))$ follows as in the proof of Lemma \ref{lem:prefix-point-mapping} and note that $\iembed$ is right adjoint.
Since $\mmng{P_D}$ is a prefixed point of $\mfunc_{P_D:\Delta}$, it follows from Lemma \ref{lem:prefix-point-mapping} that $\iembed(\mmng{P_D})$ is a prefixed point of $\sfunc_{P_D}$.
Finally, by Lemma \ref{lem:models-pfps} it is therefore a model of $D$, so $\abra{\Delta,D,G}$ is solvable.
\end{itemize}
\end{proof}

\lo{
Question: Can the complete lattice of continuous relations $\cmng{\rho}$ be embedded in the complete lattice of monotone relations $\mmng{\rho}$ in two canonical ways, analogously to (\ref{eq:monoembed}), as follows?
\[
  \mmng{\rho} \galoiS{\lomon_\rho}{\iembed_\rho} \cmng{\rho} \Galois{\jembed_\rho}{\upmon_\rho} \mmng{\rho}
\]

Consider a structure that interprets $\iota$ as a $\omega$-complete partial order (CPO) which contains a $\omega$-chain, and $\To$ as the continuous function space.
Then the inclusion map $I : \cmng{\iota \To o} \to \mmng{\iota \To o}$ is not a right adjoint;
similarly the inclusion map $\cmng{\iota \To o} \to \smng{\iota \To o}$ is not a right adjoint.

WLOG assume that the chain $\mathbb{N} \cup \makeset{\infty}$, ordered in the usual way, is contained in $\cmng{\iota} = \mmng{\iota}$.
Let $i \in \mathbb{N}$, and set $U_i := \makeset{j \in \mathbb{N} \mid j \geq i} \cup \makeset{\infty}$ and $U_\infty := \makeset{\infty}$.
(For convenience, we confuse a set with its characteristic function.)
Notice that $U_i, U_\infty \in \mmng{\iota} \mTo \mmng{o}$, and
$U_i\in \cmng{\iota} \cTo \cmng{o}$, but
$U_\infty \not\in \cmng{\iota} \cTo \cmng{o}$.
Suppose, for a contradiction, $L$ is the left adjoint of $I : \cmng{\iota} \cTo \cmng{o} \to \mmng{\iota} \mTo \mmng{o}$.
Then given $i \in \mathbb{N}$, for all $i \in \mathbb{N}$, we have
\[
L (\phi_\infty) \leq \phi_i \iff \phi_\infty \leq I(\phi_i) = \phi_i
\]
which implies $L(\phi_\infty) = \phi_\infty$, which is not in $\cmng{\iota} \cTo \cmng{o}$.
}

%% file: types.tex
\section{Refinement Type Assignment}\label{sec:types}

Having reduced the higher-order Horn clause problem to a problem about the least-fixpoint semantics of higher-order logic programs, we now consider the task of automating reasoning about such programs.
For this purpose, we look to work on refinement type systems, which are one of the most successful approaches to automatically obtaining invariants for higher-order, functional programs.
In this section, we develop a refinement type system for monotone logic programs.

\paragraph{Elimination of higher-order existentials.}
Due to monotonicity, it is possible to eliminate higher-order existential quantification from goal terms, simplifying the design of the type system.
Given any goal term of shape $\exists x\!:\!\rho.\, G$, observe that $\mmng{G}$ is a monotone function of $x$.  If there is a relation that can be used as a witness for $x$ then, by monotonicity, so too can any larger relation.
In particular, $G$ is true of some $x$ of sort $\rho$ iff $G$ is true of the universal relation of sort $\rho$, that is, the relation $u_\rho$ satisfying $u_\rho(r_1)\cdots{}(r_k) = 1$ for all $r_1,\ldots,r_k$.
Moreover, the universal relation of sort $\rho$ is itself definable by a goal term, it is just the term $U_{\rho} \coloneqq \abs{x_1\ldots{} x_k}{\truetm}$.
Consequently, $\mmng{\exists x\!:\!\rho.\, G : o} = \mmng{G[U_\rho/x]}$, in which the instance of existential quantification over relations has been eliminated by a syntactic substitution.
For the rest of the paper, we assume without loss of generality that monotone logic programs contain only existential quantification over individuals.

\subsection{Syntax}
The refinement types are built out of constraint formulas which are combined using the dependent arrow.
For the purposes of this section, we shall assume that the constraint language is closed under conjunction and closed under (well-sorted) substitution, in the sense that, for every constraint formula $\Delta,\,x\!\!:\iota \sorts \phi : o \in \aformulas$ and term $\Delta \sorts N : \iota \in \aterms$, it follows that $\phi[N/x] \in \aformulas$.

\paragraph{Types.}
The restricted syntax of goal terms allows us to make several simplifications to our refinement type system, in comparison to those in the literature.
The first is that we only allow refinement of the propositional sort $\boolsort$ and we only allow dependence on the sort of individuals $\iota$.
Formally, we define the set of type expressions according to the following grammar:
\[
  \begin{array}{crcl}
     \textsc{(Type)} & T & \Coloneqq & \boolty{\phi} \mid \dto{x}{\iota}{T} \mid T_1 \to T_2 \\
  \end{array}
\]
in which $\phi \in \aformulas$ is a constraint formula.
We make this definition under the assumption that both kinds of arrow associate to the right and we identify types up to $\alpha$-equivalence.

\paragraph{Refinement}
We restrict attention to those types that we consider to be well-formed, which is defined by a system of judgements $\Delta \types T :: \rho$, in which $\Delta$ is a sort environment, $T$ is a type and $\rho$ is a relational sort.
In case such a judgement is provable we say that $T$ \emph{refines} $\rho$.

\vspace{5pt}
\noindent
\begin{minipage}{.32\linewidth}
\begin{prooftree}
  \AxiomC{}
  \RightLabel{$\Delta \types \phi : \boolsort \in \aformulas$}
  \UnaryInfC{$\Delta \types \boolty{\phi} :: \boolsort$}
\end{prooftree}
\end{minipage}
\begin{minipage}{.32\linewidth}
\begin{prooftree}
  \AxiomC{$\Delta, x\!\!:\iota \types T :: \rho$}
  \UnaryInfC{$\Delta \types \dto{x}{\iota}{T} :: \iota \to \rho$}
\end{prooftree}
\end{minipage}
\begin{minipage}{.32\linewidth}
\begin{prooftree}
  \AxiomC{$\Delta \types T_1 :: \rho_1$}
  \AxiomC{$\Delta \types T_2 :: \rho_2$}
  \BinaryInfC{$\Delta \types T_1 \to T_2 :: \rho_1 \to \rho_2$}
\end{prooftree}
\end{minipage}

\paragraph{Type environments.}
A \emph{type environment} $\Gamma$ is a finite sequence of pairs of variables and types, $x : T$, such that the variable subjects are pairwise distinct.
We write the empty sort environment as $\epsilon$.
We place similar well-formedness restrictions on environments, using the judgement $\sorts \Gamma :: \Delta$, in which $\Gamma$ is a type environment and $\Delta$ a sort environment.
In case such a judgement is provable we say that $\Gamma$ \emph{refines} $\Delta$.

\vspace{5pt}
\noindent
\begin{minipage}{.32\linewidth}
\begin{prooftree}
\AxiomC{}
\UnaryInfC{$\sorts \epsilon :: \epsilon$}
\end{prooftree}
\end{minipage}
\begin{minipage}{.32\linewidth}
\begin{prooftree}
\AxiomC{$\sorts \Gamma :: \Delta$}
\UnaryInfC{$\sorts (\Gamma,\,x:\iota) :: (\Delta,\,x:\iota)$}
\end{prooftree}
\end{minipage}
\begin{minipage}{.32\linewidth}
\begin{prooftree}
\AxiomC{$\sorts \Gamma :: \Delta$}
\AxiomC{$\Delta \sorts T :: \rho$}
\BinaryInfC{$\sorts (\Gamma,\,x:T) :: (\Delta,\,x:\rho)$}
\end{prooftree}
\end{minipage}\\[2mm]

\noindent
Since the variable subjects of a type environment are required to be distinct, we will frequently view such an environment as a finite map from variables to types.  Thus, whenever $x$ is a subject in $\Gamma$, we will write $\Gamma(x)$ for the type assigned to $x$ and $\dom(\Gamma)$ for its set of subjects.

\paragraph{Subtype theory.}
Much of the power of refinement types is derived from the associated subtype theory, which imports wholesale the reasoning apparatus of the underlying constraint language.
Subtyping between types is the set of inequalities given by the judgement form $\types T_1 \subtype T_2$ defined inductively.\\[1mm]

\begin{center}
\begin{minipage}{.3\linewidth}
\begin{prooftree}
\AxiomC{}
\RightLabel{$\atheory \models \term{\phi} \implies \term{\psi}$}
\UnaryInfC{$\types \boolty{\phi} \subtype \boolty{\psi}$}
\end{prooftree}
\end{minipage}
\begin{minipage}{.3\linewidth}
\begin{prooftree}
\AxiomC{$\types T_1 \subtype T_2$}
\UnaryInfC{$\types \dto{x}{\iota}{T_1} \subtype \dto{x}{\iota}{T_2}$}
\end{prooftree}
\end{minipage}
\begin{minipage}{.3\linewidth}
\begin{prooftree}
\AxiomC{$\types T_1' \subtype T_1$}
\AxiomC{$\types T_2 \subtype T_2'$}
\BinaryInfC{$\types T_1 \to T_2 \subtype T_1' \to T_2'$}
\end{prooftree}
\end{minipage}
\end{center}\vspace{5pt}

\noindent
It is natural to view $\subtype$ as a preorder and useful to distinguish the extremal elements.  We write $\toptype{\rho}$ and $\bottype{\rho}$ for the families of refinement types defined inductively as follows:
\[
\begin{array}{ccc}
  \begin{array}{rcl}
    \toptype{\boolsort} &=& \boolty{\truetm} \\
    \toptype{\iota \to \rho} &=& \dto{z}{\iota}{\toptype{\rho}} \\
    \toptype{\rho_1 \to \rho_2} &=& \bottype{\rho_1} \to \toptype{\rho_2} \\
  \end{array}
  &\qquad\qquad&
  \begin{array}{rcl}
    \bottype{\boolsort} &=& \boolty{\falsetm} \\
    \bottype{\iota \to \rho} &=& \dto{z}{\iota}{\bottype{\rho}} \\
    \bottype{\rho_1 \to \rho_2} &=& \toptype{\rho_1} \to \bottype{\rho_2} \\
  \end{array}
\end{array}
\]
in which the variables $z$ are chosen to be suitably fresh.  It is clear that, by construction, $\sorts \toptype{\rho} :: \rho$ and $\sorts \bottype{\rho} :: \rho$.

\paragraph{Type assignment.}
Type assignment for goal terms $\Delta \sorts G : \sigma$ and programs $\sorts P : \Delta$ are defined by a system of judgements of the forms $\Gamma \monotypes G : T$ and $\types P : \Gamma$ respectively,
in which type environment $\Gamma$ and type $T$ are required to satisfy $\sorts \Gamma :: \Delta$ and $\Delta \sorts T :: \sigma$.
The system is defined in Figure \ref{fig:ty-assignment}.

\paragraph{Application and abstraction.}
There are two versions of each of abstraction and application, corresponding to the fact that we have chosen to emphasize the difference between dependence of a result type on an argument of individual sort and non-dependence on arguments of relational sort.

In common with other refinement type systems in the literature but unlike more general dependent type systems, our types are not closed under substitution of arbitrary terms of the programming language.
This can be reconciled with the usual rule for dependent application, in which substitution $T[N/x]$ into a type $T$ occurs, in a number of ways.
For example, in the system of \cite{rondon-et-al-pldi2008}, terms $N$ of the programming language that are substituted into a type $T$ through application are understood using uninterpreted function symbols in the logic;
in \cite{unno-kobayashi-ppdp2009}, the rule for application trivialises the substitution by requiring that $T$ does not contain the dependent variable $x$ and;
in \cite{terauchi-popl2010}, the operand $N$ is required to be a variable (which can be guaranteed when the program is assumed to be in A-normal form).
In our case, dependence can only occur at sort $\iota$ and, since the subjects of the system are goal terms, all subterms of sort $\iota$ are necessarily already terms of the constraint language, so we avoid the need for any further conditions.

\paragraph{Subtyping}
Our subsumption rule \TSub is quite standard, we just note that the fact that there are no non-trivial refinements of the base sort $\iota$ has the consequence that there is no significant advantage for the subtype judgement to refer to the type context $\Gamma$, which is why it has been formulated in this context-free way.  Similar comments can be made about avoiding the need to distinguish between base and function types in \TVar.

\paragraph{Constraints and logical constants.}
To understand the rules for typing constraint, existential, conjunctive and disjunctive formulas, it is instructive to assign a meaning to the judgement $\Gamma \types G : \boolty{\phi}$.
One should view this judgement as asserting that: in those valuations that satisfy $\Gamma$, $G$ implies $\phi$.
This statement is made precise in Lemma \ref{lem:boolty-meaning} once the semantics of types has been introduced.
Under this reading, we can view the type system as a mechanism for concluding assertions of the form ``goal formula $G$ is approximated by constraint $\phi$'' or ``constraint $\phi$ is an abstraction of goal formula $G$''.
This is a useful assertion for automated reasoning because it is relating the complicated formula $G$, which may include higher-order relation symbols whose meanings are defined recursively by the program, to the much more tractable constraint formula $\phi$, which is drawn from a (typically decidable) first-order theory.

This view helps to clarify the intuition behind the rules for assigning types to formulas headed by a constant.  In particular, for goal terms $G$ that are themselves formulas of the constraint language, the rule \TConstraint loses no information in the abstraction.
Note that we give only a rule for typing existential quantification at base sort $\iota$, which is justified by the remarks at the start of this section.
Finally, observe that the side condition on \TExists{}\! is equivalent to the condition $\atheory \models (\exists x.\,\phi) \implies \psi$ since, due to the well-formedness of the judgement, $\psi$ cannot contain $x$ freely.
We use the side condition in the given form because the constraint language may not allow existential quantification.
A similar remark may be made concerning the rule \TOr{}\! and the ability to express disjunction.

\begin{figure}[t]
\begin{center}
\hspace{-40pt}
\begin{minipage}{.4\linewidth}
\begin{prooftree}
\AxiomC{\phantom{hello}}
\LeftLabel{\TVar}
\UnaryInfC{$\Gamma_1,\, x\!\!:T,\, \Gamma_2 \monotypes x:T$}
\end{prooftree}
\end{minipage}
\begin{minipage}{.45\linewidth}
\begin{prooftree}
\AxiomC{\phantom{hello}}
\LeftLabel{\TConstraint}
\RightLabel{$\phi \in \aformulas$}
\UnaryInfC{$\Gamma \monotypes \phi : \boolty{\phi}$}
\end{prooftree}
\end{minipage}\\[3mm]
\hspace{-40pt}
\begin{minipage}{.4\linewidth}
\begin{prooftree}
\AxiomC{$\Gamma \monotypes G : T_1$}
\AxiomC{$\types T_1 \subtype T_2$}
\LeftLabel{\TSub}
\BinaryInfC{$\Gamma \monotypes G : T_2$}
\end{prooftree}
\end{minipage}
\begin{minipage}{.4\linewidth}
\begin{prooftree}
\AxiomC{$\Gamma, x:\iota \types G : \boolty{\phi}$}
\RightLabel{$\atheory \models \phi \implies \psi$}
\LeftLabel{\TExists}
\UnaryInfC{$\Gamma \types \exists x\!\!:\!\!\iota.\,G : \boolty{\psi}$}
\end{prooftree}
\end{minipage}\\[3mm]
\hspace{-40pt}
\begin{minipage}{.45\linewidth}
\begin{prooftree}
\AxiomC{$\Gamma \types G : \boolty{\phi}$}
\AxiomC{$\Gamma \types H : \boolty{\psi}$}
\LeftLabel{\TAnd}
\BinaryInfC{$\Gamma \types G \wedge H : \boolty{\phi \wedge \psi}$}
\end{prooftree}
\end{minipage}
\begin{minipage}{.45\linewidth}
\begin{prooftree}
\AxiomC{$\Gamma \types G : \boolty{\phi}$}
\AxiomC{$\Gamma \types H : \boolty{\psi}$}
\LeftLabel{\TOr}
\RightLabel{$\begin{array}{c}\atheory \models \phi \implies \chi \\\atheory \models \psi \implies \chi\end{array}$}
\BinaryInfC{$\Gamma \types G \vee H : \boolty{\chi}$}
\end{prooftree}
\end{minipage}\\[3mm]
\hspace{-40pt}
\begin{minipage}{.45\linewidth}
\begin{prooftree}
\AxiomC{$\Gamma, x\!\!:T_1 \monotypes G : T_2$}
\LeftLabel{\TAbsR}
\UnaryInfC{$\Gamma \monotypes \abs{x\!\!:\!\!\rho}{G} : T_1 \to T_2$}
\end{prooftree}
\end{minipage}
\begin{minipage}{.45\linewidth}
\begin{prooftree}
\AxiomC{$\Gamma \monotypes G : T_1 \to T_2$}
\AxiomC{$\Gamma \monotypes H : T_1$}
\LeftLabel{\TAppR}
\BinaryInfC{$\Gamma \monotypes \term{G H} : T_2$}
\end{prooftree}
\end{minipage}\\[3mm]
\hspace{-40pt}
\begin{minipage}{.4\linewidth}
\begin{prooftree}
\AxiomC{$\Gamma, x\!\!:\iota \monotypes G : T$}
\LeftLabel{\TAbsI}
\UnaryInfC{$\Gamma \monotypes \abs{x\!\!:\!\!\iota}{G} : \dto{x}{\iota}{T}$}
\end{prooftree}
\end{minipage}
\begin{minipage}{.45\linewidth}
\begin{prooftree}
\AxiomC{$\Gamma \monotypes G : \dto{x}{\iota}{T}$}
\AxiomC{$\sorts \Gamma :: \Delta$}
\LeftLabel{\TAppI}
\BinaryInfC{$\Gamma \monotypes \term{G N} : T[N/x]$}
\end{prooftree}
\end{minipage}
\end{center}
\caption{Type assignment for goal terms.}\label{fig:ty-assignment}
\end{figure}

To programs $\sorts P : \Delta$, of shape $x_1 = G_1,\,\ldots,\,x_m=G_m$, we assign type environments according to the rule:
\begin{prooftree}
\AxiomC{$\Gamma \types G_1 : \Gamma(x_1) \qquad\cdots\qquad \Gamma \types G_m : \Gamma(x_m)$}
\LeftLabel{\TProg}
\UnaryInfC{$\types x_1 = G_1, \quad\ldots,\quad x_m = G_m : \Gamma$}
\end{prooftree}
Once we have defined the relational semantics of types, in which a type environment $\Gamma$ is interpreted as a valuation $\rmmng{\Gamma}$, the soundness of this rule will guarantee that $\rmmng{\Gamma}$ is a prefix point of $\mfunc_{P:\Delta}$, and hence is an over-approximation of $\mmng{P}$.

\begin{example}
Using \TProg{}\! the program in Example \ref{ex:iter-horn-program} can be assigned the type environment $\Gamma_I$ from the introduction.
For example, the type of $\Iter$ can be justified from the following judgements.
First, for the subterms $n \leq 0$ and $m = 0$ in the body of $\Iter{}$ we can apply the \TConstraint rule to immediately derive the judgements:
\[
  \Gamma' \types n \leq 0 : \boolty{n \leq 0} \quad\text{ and }\quad \Gamma' \types m = s : \boolty{m = s}
\]
in which $\Gamma'$ is the type environment $\Gamma,\, f\!\!:\dto{x}{\intty}{\dto{y}{\intty}{\dto{z}{\intty}{\boolty{0 < x \implies y < z}}}},\, s\!\!:\intty,\, n\!\!:\intty,\, m\!\!:\intty$.
Then using \TAnd{}\!, we can derive the judgement:
$
  \Gamma' \types n \leq 0 \wedge m = s : \boolty{n \leq 0 \wedge m = s}
$.
Moreover, since also $\mathsf{ZLA} \models n \leq 0 \wedge m = s \implies 0 \leq s \implies n \leq m$, it follows from the subsumption rule \TSub{}\! that the judgement:
\begin{equation}
  \Gamma' \types n \leq 0 \wedge m = s : \boolty{0 \leq s \implies n \leq m}
\end{equation}
is also derivable.
Next, consider the subterm $\term{\Iter{} f s (n-1) p}$.
Observe that $\types \boolty{z = x + y} \subtype \boolty{0 < x \implies y < z}$ so, using \TVar and \TSub we can assign to $f$ the type
\[
  \dto{x}{\intty}{\dto{y}{\intty}{\dto{z}{\intty}{\boolty{0 < x \implies y < z}}}}
\]
Hence, it follows from the \TAppR{}\! and \TAppI{}\! rules that we can derive the judgement:
$
  \Gamma' \types \term{\Iter{} f s (n-1) p} : \boolty{0 \leq s \implies n-1 \leq p}
$.
To the term $\term{f n p m}$ we can assign the type $\boolty{0 < n \implies p < m}$.
Hence, by \TAnd{}\! and then \TExists{}\! we can conclude the judgement:
\begin{equation}
  \Gamma' \types \exists p.\,0 < n \wedge \term{\Iter{} f s (n-1) p} \wedge{} \term{f n p m} : \boolty{0 \leq s \implies n \leq m}
\end{equation}
since $\mathsf{ZLA} \models (\exists p.\, 0 < n \wedge (0 \leq s \implies n-1 \leq p) \wedge (0 < n \implies p < m)) \implies 0 \leq s \implies  n \leq m$.
Using (1), (2) and \TOr{} we can therefore derive an overall type for the body of the definition of $\Iter$ as $\boolty{0 \leq s \implies n \leq r}$.
The desired type follows from applications of \TAbsI{} and \TAbsR{}\!.
\end{example}

\subsection{Semantics}
We ascribe two meanings to types.
The first is the usual semantics in which types are some kind of set and, as is typical, such sets have the structure of an order ideal.
The second is specific to our setting and exploits the fact that every type refines a relational sort.
The second semantics assigns to each type a specific relation.
This makes later developments, such as the demonstration of the soundness of type assignment, simpler and also makes a link to the notion of symbolic model from the first-order case.

\paragraph{Ideal semantics of types.}
The \emph{ideal semantics} of refinement types is defined so as to map a well-formed typing sequent $\Delta \types T :: \rho$ and appropriate valuation $\alpha$ to a subset $\mmng{T}(\alpha) \subseteq \mmng{\rho}$.
\[
\begin{array}{rcl}
  \mmng{\Delta \sorts \boolty{\phi} :: \boolsort}(\alpha) &=& \{0, \smng{\Delta \types \phi : \boolsort}(\alpha) \} \\
  \mmng{\Delta \sorts \dto{x}{\iota}{T} :: \iota \to \rho}(\alpha) &=& \{ r  \mid \forall n \in A_\iota.\; r(n) \in \mmng{\Delta,x:\iota \sorts T :: \rho}(\alpha[x \mapsto n])\} \\
  \mmng{\Delta \sorts T_1 \to T_2 :: \rho_1 \to \rho_2}(\alpha) &=& \{ f  \mid \forall r \in \mmng{\Delta \types T_1 :: \rho_1}(\alpha).\, f(r) \in \mmng{\Delta \sorts T_2 :: \rho_2}(\alpha)\} \\
\end{array}
\]
We are now in a position to make precise the remark following the definition of type assignment.

\begin{lemma}\label{lem:boolty-meaning}
$\mmng{G}(\alpha) \in \mmng{\boolty{\phi}}(\alpha)\;$ iff $\;\alpha \models G \implies \phi$.
\end{lemma}

\noindent
The soundness of type assignment will guarantee that, additionally, $\Gamma \types G : \boolty{\phi}$ implies that $\mmng{G}(\alpha) \in \mmng{\boolty{\phi}}(\alpha)$ for all appropriate $\alpha$.

\paragraph{Relational semantics of types.}
We will also consider a monotone \emph{relational} semantics of types.
Fix a sort environment $\Delta$.  Given an appropriate first-order valuation $\alpha$,  we can associate relations $\rmmng{\Delta \sorts T :: \rho}(\alpha) \in \mmng{\rho}$ to each well formed type $\Delta \sorts T : \rho$.
\[
\begin{array}{rcll}
  \rmmng{\Delta \types \boolty{\phi} :: \boolsort}(\alpha) &=& \smng{\Delta \types \phi : \boolsort}(\alpha) &\\
  \rmmng{\Delta \types \dto{x}{\iota}{T} :: \iota \to \rho}(\alpha)(n) &=& \rmmng{\Delta, x:\iota \sorts T :: \rho}(\alpha[x \mapsto n]) &\\
  \rmmng{\Delta \types T_1 \to T_2 :: \rho_1 \to \rho_2}(\alpha)(r) &=&
      \rmmng{\Delta \sorts T_2 :: \rho_2}(\alpha) & \text{if $r \subseteq \rmmng{\Delta \sorts T_1 :: \rho_1}(\alpha)$} \\
  \rmmng{\Delta \types T_1 \to T_2 :: \rho_1 \to \rho_2}(\alpha)(r) &=& \rmmng{ \epsilon \sorts \toptype{\rho_2} :: \rho_2}(\emptyset) & \text{otherwise} \\
\end{array}
\]
It follows straightforwardly from the definitions that  $\rmmng{T}(\alpha)$ is a monotone relation.

\paragraph{Symbolic models.}
Additionally, we can consider a type environment under a relational interpretation.  We define, for all $x \in \dom(\Delta)$, $\rmmng{\sorts \Gamma :: \Delta}(x) = \rmmng{\Gamma(x)}(\emptyset)$.  This object is a monotone valuation and hence, following the results of Section \ref{sec:red-to-eval}, this semantics makes precise the idea of a type environment $\Gamma$ as a (candidate) \emph{symbolic model}, that is, a finite representation of a model designed for the purposes of automated reasoning.

\paragraph{Relationship.}
The ideal and the relational semantics are closely related: the relational interpretation $\rmmng{T}(\alpha)$ of a type $T$ is the largest element of the ideal $\mmng{T}(\alpha)$.
Thus we see also that $\mmng{T}(\alpha)$ is principal.
For each relation $r \in \mmng{\rho}$ let us write $\pideal_\rho r$ for the downward closure of $\{r\}$ in $\mmng{\rho}$ (we will typically omit the subscript).

\begin{lemma}\label{lem:types-are-pideals}
For all types $\Delta \sorts T :: \rho$ and $\alpha \in \smng{\Delta}$, $\mmng{T}(\alpha) = \pideal\rmmng{T}(\alpha)$.
\end{lemma}

\subsection{Soundness}

We return to the discussion of type assignment by defining semantic judgements paralleling those for the type system.
First let us make a preliminary definition.
Let $\sorts \Gamma :: \Delta$, then we say that valuation $\alpha \in \mmng{\Delta}$ \emph{satisfies} $\Gamma$ just if, for all $x \in \dom(\Delta)$, $\alpha(x) \in \mmng{\Gamma(x)}$.
We write the following judgement forms:
\[
  \models T_1 \subtype T_2 \qquad\qquad\qquad \Gamma \models G : T \qquad\qquad\qquad \models P : \Gamma
\]
In the first, we assume that $\Delta \sorts T_1 : \rho$ and $\Delta \sorts T_2 : \rho$ are types of the same sort.
Then the meaning of $\models T_1 \subtype T_2$ is that, for all $\alpha \in \mmng{\Delta}$, $\mmng{T_1}(\alpha) \subseteq \mmng{T_2}(\alpha)$.
In the second, we assume that $\sorts \Gamma :: \Delta$, $\Delta \sorts G : \rho$ and $\Delta \sorts T :: \rho$.
Then the meaning of $\Gamma \models G : T$ is that, for all valuations $\alpha \in \mmng{\Delta}$, if $\alpha$ satisfies $\Gamma$ then $\mmng{G}(\alpha) \in \mmng{T}(\alpha)$.
Finally, in the third, we assume that $\sorts P : \Delta$ and $\sorts \Gamma :: \Delta$.
Then the meaning of $\models P : \Gamma$ is that $\mmng{P}$ satisfies $\Gamma$.

\begin{theorem}[Soundness of Type Assignment]\label{thm:ty-ass-soundness}
\leavevmode \\
\begin{enumerate}[(i)]
  \item If \;$\types T_1 \subtype T_2$\; then \;$\models T_1 \subtype T_2$.\qquad
  \item If \;$\Gamma \types G : T$\; then \;$\Gamma \models G : T$.\qquad
  \item If \;$\types P : \Gamma$\; then \;$\models P : \Gamma$.
\end{enumerate}
\end{theorem}
\begin{proof}[Proof sketch]
The first and second claims are proven by straightforward inductions on the relevant judgements.
In the third case it follows from (ii) and the definition of \TProg{} that $\rmmng{\Gamma}$ satisfies $\Gamma$ and is therefore a prefixed point of $\mfunc_{P:\Delta}$.
Since $\mmng{P}$ is the least such, the result follows.
\end{proof}

This allows for a sound approach to solving systems of higher-order constrained Horn clauses.
Given an instance of the problem $\abra{\Delta,D,G}$, we construct the corresponding logic program $\abra{\Delta,P_D,G}$.
If there is a type environment $\sorts \Gamma :: \Delta$ such that $\types P_D : \Gamma$ and $\Gamma \types G : \boolty{\falsetm}$ then it follows from Theorem \ref{thm:ty-ass-soundness} that, for each model $A$ of the background theory, $\rmmng{\Gamma}$ (with constants interpreted with respect to $A$) is a valuation that satisfies $D$ but refutes $G$.
The approach is, however, incomplete.
Consider the following instance, adapted from an example of \cite{unno-et-al-popl2013} showing the incompleteness of refinement type systems for higher-order program verification.
\begin{example}\label{ex:incompleteness}
The higher-order constrained Horn clause problem $\abra{\Delta,D,G}$ specified by, respectively:
\[
  \begin{array}{ccc}
    \begin{array}{l}
      \mathit{Leq} : \intsort \to \intsort \to \boolsort \\
      \mathit{Holds} : (\intsort \to \boolsort) \to \intsort \to \boolsort
    \end{array}
  &
  \begin{array}{l}
    \forall i\ j.\; i \leq j \implies \mathit{Leq}\ i\ j \\
    \forall p\ n.\; p\ n \implies \mathit{Holds}\ p\ n
  \end{array}
  &
  \exists i.\, \mathit{Holds}\ (\mathit{Leq}\ i)\ (i-1)
  \end{array}
\]
is solvable in the theory of integer linear arithmetic ($\mathsf{ZLA}$), since there is no integer $i$ smaller than its predecessor.
However, the logic program $P_D$ defined as $\mathit{Leq} = \abs{ij}{i \leq j},\; \mathit{Holds} = \abs{pn}{pn}$ is not typable.
This is because if there were a type environment $\sorts \Gamma :: \Delta$ for which $\types P_D : \Gamma$ and $\Gamma \types G : \boolty{\falsetm}$, it would have shape:
\[
    \mathit{Leq} : \dto{i}{\intty}{\dto{j}{\intty}{\boolty{\chi}}} \qquad
    \mathit{Holds} : (\dto{x}{\intty}{\boolty{\phi}}) \to \dto{n}{\intty}{\boolty{\psi}}
\]
for some formulas $i\!\!:\intsort,j\!\!:\intsort \sorts \chi : o$ and $x\!\!:\intsort \sorts \phi : o$ and $n\!\!:\intsort \sorts \psi : o$.
These formulas would necessarily satisfy the following conditions:
\begin{inparaenum}
  \item $\mathsf{ZLA} \models i \leq j \implies \chi$,
  \item $\mathsf{ZLA} \models \phi[n/x] \implies \psi$,
  \item $\mathsf{ZLA} \models \chi[x/j] \implies \phi$, and
  \item $\mathsf{ZLA} \models (\exists i.\,\phi[i-1/n]) \implies \falsetm$,
\end{inparaenum}
implied by the definition of the type system.
It follows from (1) and (3) that also $\mathsf{ZLA} \models i \leq x \implies \phi$.
Since $\phi$ does not contain $i$ freely, it follows that this is equivalent to $\mathsf{ZLA} \models (\exists i.\, i \leq x) \implies \phi$, itself equivalent to $\mathsf{ZLA} \models \phi$.  However, this contradicts (4), so there can be no such type assignment.
\end{example}
\noindent
It seems likely that, to obtain a sound and (relatively) complete approach, one could adapt the development described in \cite{unno-et-al-popl2013}, at the cost of complicating the system a little.

%% file: automation.tex
\subsection{Automation}\label{sec:automation}
This approach to solving the higher-order constrained Horn clause problems relies on finding a refinement type environment to act as a witness.
It is well understood that, for refinement type systems following a certain pattern, typability (i.e. the existence of a type environment satisfying some properties) can be reduced to first-order constrained Horn clause solving \cite{terauchi-popl2010,BjornerMR12,jhala-et-al-cav2011}.
Hence, the search for a witnessing environment can be automated using standard techniques
\ifsupp
(an explicit definition is included in Appendix~\ref{sec:apx-inference}).
\else
(an explicit definition is included in the %
supplementary materials).
\fi

In order to check that the approach is feasible, we have implemented a prototype tool and used it to automatically verify a few small systems of clauses.
The tool is written in Haskell and uses the Parsec library for parsing input in a mathematical syntax using unicode characters or ascii equivalents. Output is either a similar format or conforms to SMT-LIB in a manner that lets Z3 \cite{DeMoura:2008:ZES:1792734.1792766} solve the system of first order clauses.

The test cases are obtained from that subset of the functional program verification problems given in \cite{kobayashi-et-al-pldi2011} which do not make local assertions, re-expressed as higher-order Horn clause problems according to the method sketched by example in Section \ref{sec:intro}.
The prototype\footnote{\changed[lo]{A basic web interface to the prototype tool is available at \url{http://mjolnir.cs.ox.ac.uk/horus}}.
}
and the exact suite of test cases is available at \url{http://github.com/penteract/HigherOrderHornRefinement}.
In all
but one
of the examples the prototype takes around 0.01s to transform the system of clauses and Z3 takes around 0.02s to solve the resulting first-order system.
The remaining example, named \texttt{neg}, suffers from our choice of refinement type system.
It is solved by the system described in \emph{loc.~cit.}, which allows type-level intersection, but cannot be solved using our approach.

%% file: extension.tex
\subsection{Expressibility of type assertions}

It is possible to express the complement of a type using a goal formula.  
For example, according to the forgoing semantics, the type 
\[
T \coloneqq (\dto{x}{\iota}{\dto{y}{\iota}{\boolty{x \equiv y \mathrel{\mathsf{mod}} 2}}}) \to \dto{z}{\iota}{\boolty{z \neq 0}}
\]
represents the set $\mmng{T}(\emptyset)$ of all monotone relations of sort $(\iota \to \iota \to o) \to \iota \to o$ that relate all parity-preserving inputs to non-zero outputs.
The complement of this set of relations can be defined (in the monotone semantics) by the goal term $\com(T) \coloneqq \abs{z}{z (\abs{xy}{x \equiv y \mathrel{\mathsf{mod}} 2})\ 0}$
which classifies such relations according to whether or not they relate the particular parity-preserving relation $\abs{xy}{x \equiv y \mathrel{\mathsf{mod}} 2}$ to $0$.
The mode of definition we have in mind is the following.

\paragraph{Goal definability.}
We say that a relation $r \in \mmng{\rho}$ is \emph{goal term definable} (more briefly \emph{$G$-definable}) just if there exists a closed goal term $\types H : \rho$ and $\mmng{H}(\emptyset) = r$.  We say that a class of relations is $G$-definable just if the characteristic predicate of the class is $G$-definable.\\

Returning to the example, if a given relation $r$ does not relate $\abs{xy}{x \equiv y \mathrel{\mathsf{mod}} 2}$ to $0$, then $r$ is not a member of $\mmng{T}(\emptyset)$.
On the other hand, if $r$ is not a member $\mmng{T}(\emptyset)$, then there is some parity-preserving relation $s$ that is related by $r$ to $0$.
Since $r$ is monotone, it follows that $r$ also relates all relations larger than $s$ to $0$.
Since $\abs{xy}{x \equiv y \mathrel{\mathsf{mod}} 2}$ is the  largest of the parity-preserving relations, it follows that $r$ relates $\abs{xy}{x \equiv y \mathrel{\mathsf{mod}} 2}$ to $0$.
Hence the set 
\[
\{ r \in \mmng{(\iota \to \iota \to o) \to \iota \to o} \mid r \notin \mmng{T}(\emptyset) \}
\] 
is $G$-definable by the goal term $\com(T)$.

\paragraph{Definability of type complements and the relational semantics.}
Any $G$-definition $\com(T)$ of the complement of $T$ is intertwined with a $G$-definition of the largest element of $T$: to understand when a relation $r$ is not in the type $T_1 \to T_2$ you must understand when there exists a relation $s \in \mmng{T_1}(\emptyset)$ such that $r(s)$ is not in $\mmng{T_2}(\emptyset)$, i.e. when the property $\mmng{\com(T)}(\emptyset)(r(s))$ holds.
However, since $\mmng{\com(T)}(\emptyset)$ and $r$ are monotone, if $\mmng{\com(T)}(\emptyset)(r(s))$ holds then  $\mmng{\com(T)}(r(s'))$ holds for any $s' \supseteq s$.
Consequently, there exists an $s \in \mmng{T_1}(\emptyset)$ satisfying the property iff the largest element of $\mmng{T_1}(\emptyset)$ satisfies the property.
This leads to the definitions by mutual induction on type:
\[
  \begin{array}{cc}
    \arraycolsep=1.4pt
    \begin{array}{rcl}
      \com(\boolty{\phi}) &=& \abs{z}{z \wedge \neg\phi} \\
      \com(\dto{x}{\iota}{T}) &=& \abs{z}{\exists  x\!:\!\iota.\,\com(T)\,(z\,x)} \\
      \com(T_1 \to T_2) &=& \abs{z}{\com(T_2)\,(z\,\lar(\falsetm)(T_1))}
    \end{array}
    &
    \arraycolsep=1.4pt
    \begin{array}{rcl}
      \lar(G)(\boolty{\phi}) &=& G \vee \phi \\
      \lar(G)(\dto{x}{\iota}{T}) &=& \abs{z}{\lar(G)(T[z/x])} \\
      \lar(G)(T_1 \to T_2) &=& \abs{z}{\lar(G \vee{} \com(T_1)\,z)(T_2)}
    \end{array}
  \end{array}
\]
in which $\com(T)$ is a $G$-definition of the complement of the class $\mmng{T}(\emptyset)$.
The definition of $\lar$ is parametrised by a goal formula representing domain conditions that are accumulated during the analysis of function types.
It follows that $\lar(\falsetm)(T)$ is a $G$-definition of the largest element of the class $\mmng{T}(\emptyset)$.

\begin{lemma}\label{lem:G-defns}
For any closed type $\sorts T :: \rho$:
\begin{enumerate}[(i)]
  \item The goal term $\com(T)$ is a $G$-definition of the class $\{r \in \mmng{\rho} \mid r \notin \mmng{T}(\emptyset)\}$.
  \item The goal term $\lar(\falsetm)(T)$ is a $G$-definition of the relation $\rmmng{T}(\emptyset)$.
\end{enumerate}
\end{lemma}

This gives a sense in which we can express the higher-type property $G : T$ within the higher-order constrained Horn clause framework, namely as the goal formula $\com(T)\ G$ defining the negation of the type assertion (recall that in the higher-order Horn clause problem we aim to refute the negation of property to be proven, expressed as a goal formula).
Furthermore, we can use this result in order to justify an extension of the syntax of higher-order constrained Horn clauses with a new kind of \emph{type-guarded} existential quantification.
This allows us to state goal formulas like $\exists x\!:\!T. G$, i.e. there exists a relation in (the set defined by) refinement type $T$ that moreover satisfies $G$.
The full development is contained in the anonymous supplementary materials.

%% file: related-work.tex
\section{Related Work}\label{sec:related-work}

\paragraph{Constrained Horn-clause solving in first-order program verification}
Our motivation comes mainly from the use of constrained Horn clauses to express problems in the verification of first-order, imperative programs.
The papers of \citet{BjornerMR12} and \citet{bjorner-et-al-flc2015} argue the case for the approach and provide a good overview of it.
One of the best recommendations for the approach is the selection of highly efficient solvers that are available, such as \cite{grebenschchikov-et-al-TACAS2012,Hoder-et-al-cav2011,gurfinkel-et-al-cav2015}, which we exploit in this work  as part of the automation of our prototype solver for higher-order clauses.

\paragraph{Higher-order logic programming}
Work on higher-order logic programming is typically concerned with programming language design and implementation.
Consequently, one of the main themes of the work discussed in the following is that of (i) finding a good semantics for higher-order logic programming languages, and one of the central criteria for a good semantics is that (ii) it lends itself well to developing techniques for enumerating answers to queries.
In contrast to (i), our work is about a particular satisfiability problem of logic, namely the existence of a model of a formula satisfying certain criteria (modulo a background theory).
In higher-order logic programming \changed[lo]{(e.g.~\citet{nadathur-miller-jacm1990,chen-et-al-jlp1993,CharalambidisHRW13})} it often does not make sense, a priori, to ask about the existence of models because the semantics of the logic programming language is fixed once and for all as part of its definition.
One can ask whether the set of answers to a program query is empty, but to understand this as a logical question about the models of a formula is only possible through a result such as our contribution in Section \ref{sec:red-to-eval}.
In contrast to (ii), our goal is to develop techniques, such as the type system in Section \ref{sec:types}, for the dual problem, namely to show that a Horn formula is satisfiable (recall that we require the existence of a model refuting the goal $G$, and the goal is the negation of a Horn clause $G \implies \falsetm$).
Query answering (unsatisfiability) is recursively enumerable, whereas satisfiability, in particular where there is a background theory, is typically harder.

\lo{Paragraph about intensional higher-order logic programming removed.}

Recent work on extensional semantics for higher-order logic programming started with \citet{Wadge91} and continued with, for example, \citet{CharalambidisHRW13,charalambidis-et-al-tplp2014}.
Wadge was the first to observe that relational variables appearing as arguments in the head of clauses is problematic.
This line of work gives a denotational semantics to higher-order logic programs and so is very closely related to the monotone semantics of logic programs that we use in Section \ref{sec:red-to-eval}, except that there is no treatment of constraint theories (discussed in the following paragraph).
Unlike our work, \citet{CharalambidisHRW13} are very careful to ensure the algebraicity of their domains so that they can build a sophisticated system of query answering based on enumerating compact elements.
If we were to extend our work to encompass a treatment of counterexamples, then we would want to exploit algebraicity in a similar manner.
Their work does not make any connection between their denotational semantics of higher-order logic programs and satisfiability in standard higher-order logic;
it seems likely that our result could be adapted to their setting.

\paragraph{Higher-order constraint logic programming.}
Another way in which all of the foregoing work differs from our own is that, in the work we have mentioned so far, there is no treatment of constraints.
Whilst first-order constraint logic programming is a very well developed area (an old, but good survey is \cite{jaffar-maher-jlp1994}), there is very little existing work on the higher-order extension.
\citet{lipton-nieva-tlca2007} give a Kripke semantics for a $\lambda$Prolog-like language extended with a generic framework for constraints.
In contrast to our work, the underlying higher-order logic is intuitionistic and the precise notion of model is bespoke to the paper.

\paragraph{Interpretations of higher-order logic}
Under the standard semantics, an interpretation of a higher-order theory consists of a choice of universe $A_\iota$ in which to interpret the sort of individuals $\iota$ and an interpretation of  the constants of the theory.
In particular, the interpretation $\smng{\sigma_1 \to \sigma_2}$ of any arrow sort $\sigma_1 \to \sigma_2$ is fixed by the choice of $A_\iota$.
In Henkin (or general) semantics, an interpretation consists of all of the above but, additionally, also a choice of interpretation for each of the infinitely many arrow sorts (under some natural restrictions regarding definability of elements).
For example, there are Henkin interpretations in which the collection $\mng{\intsort \to \boolsort}$ does not contain all sets of integers.
We choose to frame the higher-order constrained Horn clause problem using standard semantics because it is already established in verification.
For example, when monadic second order logic (MSO) is used to express verification problems on transition systems,
the second-order variables range over all sets of the states\footnote{Or all finite sets of states in the case of weak MSO, but in both cases the domain is fixed by the setting.  A Henkin semantics formulation of the problem would allow for a solution to the satisfiability problem to specify its own domain for the second order variables.}.
Consequently, the standard semantics seems the most appropriate starting point for work on higher-order constrained Horn clauses in verification.

\paragraph{Automated verification of functional programs}
Two of the most well-studied approaches to the automated verification of functional programs are based on higher-order model checking \cite{ong-lics2006,kobayashi-ong-lics2009,kobayashi-jacm2013} and refinement types \cite{rondon-et-al-pldi2008,vazou-et-al-icfp2015,jhala-et-al-cav2011,zhu-jagannathan-vmcai2013,unno-et-al-popl2013}.
A method to verify higher-order programs more directly using first-order constrained Horn clauses has been suggested by \citet*{bjorner-et-al-arxiv2013}.

In higher-order model checking, the problem of verifying a higher-order program is reformulated as the problem of verifying a property of the tree generated by a higher-order recursion scheme, see e.g. \cite{kobayashi-et-al-pldi2011}.
The approach has the advantage of being based around a natural, decidable problem, which is an attractive target for the construction of efficient solvers \cite{ramsay-et-al-popl2014,broadbent-et-al-icfp2013,broadbent-kobayashi-csl2013}.
By contrast, we propose to investigate higher-order program verification based around the higher-order constrained Horn clause problem.
Although this problem is generally undecidable, it has the advantage of being able to express (background theory) constraints directly and so has the potential to be a better setting in which to search for higher-order program invariants.

In approaches based on refinement types,
a type system is used to the reduce the problem of finding an invariant for the higher-order program, to finding a number of first-order invariants of the ground-type data at certain program points, which can often be expressed as first-order constrained Horn clause solving.
As exemplified by the LiquidHaskell system of \citet{Vazou:2014:RTH:2628136.2628161}, one advantage of using a type system directly is that it can very naturally encompass all the features of modern programming languages.
We have not addressed the problem of how best to frame a higher-order program verification problem using higher-order clauses (excepting our motivating sketch in the introduction) but it does not seem as clear as for approaches using refinement types.
On the other hand, the reduction to first-order invariants that underlies refinement type approaches has a cost in expressibility.
In principle, it is possible to overcome this deficiency, for example by employing types in which higher-order invariants can be encoded as first-order statements of arithmetic \cite{unno-et-al-popl2013};
we mention also the scheme of \citet{bjorner-et-al-arxiv2013} which proposes to view higher-order invariants as first-order statements about closures (encoded as data structures).
However, in both cases it seems plausible that working directly in higher-order logic may lead to the development of more transparent \changed[lo]{and generic} techniques.
To benefit from this would necessitate a move to a different  technology from our system of refinement types used for solving.

%% file: conclusion.tex
\section{Conclusion and future work}\label{sec:conclusion}

In this work, we have presented our notion of higher-order constrained Horn clauses and the first foundational results, with an emphasis on making connections to existing work in the verification of higher-order programs.
By analogy with the situation for first-order program verification,
we believe that higher-order constrained Horn clauses can be an attractive, programming-language independent setting for developing automated techniques for the verification of higher-order programs.
Let us conclude by giving some more justification to this belief through a discussion of future work.

We have shown, in Example \ref{ex:incompleteness}, that our method for solving is bound by the same limitations as typical refinement type systems in the literature.
However, there is a lot of scope to develop new approaches to solving which may not suffer in the same way.
A general result in this direction would be to show that these limitations are not intrinsic to the higher-order constrained Horn clause problem.
For example, it seems plausible that higher-order clauses are expressive (in the sense of Cook) for suitable Hoare logics over higher-order programs, mirroring the case at first-order \cite{blass-gurevich-ctl1987, bjorner-et-al-flc2015}.
A specific technique for solving,
which we plan to pursue,
is an approach for reducing higher-order clauses to first-order clauses with datatypes, using the ideas of \citet{bjorner-et-al-arxiv2013} (which is itself in the spirit of \changed[lo]{Reynolds'} defunctionalisation \cite{reynolds-pacm1972}).

In defunctionalisation, as in our method of Section \ref{sec:types}, the goal is to reduce the problem to that of first-order clauses by doing some reasoning about the behaviour of higher-order functions.
Without further investigation, it is unclear whether this reduction should happen at the level of programs (as is the case in e.g. \cite{bjorner-et-al-arxiv2013}) or at the level of a higher-order intermediate representation such as higher-order constrained Horn clauses (as we have suggested in the introduction).
However, if one believes that the reasoning involved in such a reduction is generic, in the sense of being essentially the same whether the programs are written in Haskell, or ML, or Javascript, then this suggests that it is worthwhile to investigate whether this reasoning can be done efficiently on the intermediate representation, rather than have the reasoning re-implemented in each different analysis of each different higher-order programming language.

Finally, we make the observation that higher-order constraints may be useful even in the verification of first-order procedures.
For example, in refinement type systems, a refinement of a base type is typically a predicate on values of the type.
Therefore, it seems reasonable that a refinement of a type \emph{constructor} should be a (Boolean-valued) function on predicates, i.e.~a higher-order relation.
It would be interesting to develop a type inference algorithm that can reason about refinements of type constructors (a concern that is orthogonal to the existence of higher-order procedures) using higher-order constraints, and to understand the connections with the work of \citet*{vazou-et-al-esop2013}.

%% file: apx-sec4.tex
\section{Supplementary material for Section~4}

\subsection{Proof of Lemma~\ref{lem:embed-mngs}}

We first consider the special case in which $G$ is existential quantification and rephrase the latter claim equivalently as $\lomon \circ \smng{\exists} \subseteq \mmng{\exists} \circ \lomon$ (which is easier to prove).

\begin{lemma}\label{lem:exists-mngs}
  For sorts $\sigma$ (either $\iota$ or some $\rho$), \ \ $\mexistsfn_\sigma \circ \upmon_{\sigma \to o} \; \subseteq \; \existsfn_\sigma \; \subseteq \; \mexistsfn_\sigma \circ \lomon_{\sigma \to o}$.
\end{lemma}
\begin{proof}
  In case $\sigma$ is $\iota$, by definition $\mexistsfn_\iota = \existsfn_\iota$ and $\upmon_{\iota \to o}(s) = s = \lomon_{\iota \to o}(s)$, so the result is clear.  Otherwise, $\sigma$ is some relational sort $\rho$ and we observe that both of the following are true for all $s \in \smng{\rho \to o}$:
 \begin{itemize}[(i)]
 \item for all $r \in \mmng{\rho}$: $\upmon_{\rho \to o}(s)(r) = 1$ implies $s(\jembed_\rho(t)) = 1$
 \item for all $t \in \smng{\rho}$: $s(t) = 1$ implies $\lomon_{\rho \to o}(s)(\lomon_\rho(t)) = 1$
 \end{itemize}
 Hence, (i) witnesses $r$ to the satisfiability of $\upmon(s)$ can be mapped to witnesses $\jembed(t)$ to the satisfiability of $s$ and (ii) witnesses $t$ to the satisfiability of $s$ can be mapped to witnesses $\lomon(t)$ to the satisfiability of $\lomon(s)$; thus proving the lemma.
 To see that (i) is true we just observe that if $\upmon_{\rho \to o}(s)(r) = 1$, then $s(\jembed_\rho(r)) = 1$ by definition.  To see that (ii) is true, we reason as follows.  If $s(t) = 1$, then since $\iembed \circ \lomon$ is inflationary, also $\iembed_{\rho \to o}(\lomon_{\rho \to o}(s))(t) = 1$.  Hence, by definition, $\lomon_{\rho \to o}(s)(\lomon_{\rho}(t)) = 1$.
\end{proof}

\begin{lemn}[\ref{lem:embed-mngs}]
For all goal terms $\Delta \sorts G : \rho$,\ \
$\jembed_{\rho} \circ \mmng{G} \circ \upmon_{\Delta} \; \subseteq \; \smng{G} \; \subseteq \; \iembed_{\rho} \circ \mmng{G} \circ \lomon_{\Delta}$.
\end{lemn}
\begin{proof}
  The proof of the inclusion $\jembed \circ \mmng{G} \circ \upmon \subseteq \smng{G}$ is by induction on the sorting judgement for goal terms.  We give here only the more interesting cases:
  \begin{itemize}
    \item If $\Delta \sorts x : \rho$ then:
    \[
      \jembed(\mmng{x}(\upmon(\alpha))) = \jembed(\upmon(\alpha(x))) \subseteq \alpha(x) = \smng{x}(\alpha)
    \]
    by definition and the fact that $\jembed \circ \upmon$ is deflationary.
    \item If $\Delta \sorts \phi : o$ with $\phi$ a formula of the constraint language, then:
    \[
      \jembed(\mmng{\phi}(\upmon(\alpha))) = \smng{\phi}(\upmon(\alpha)) = \smng{\phi}(\alpha)
    \]
    by definition and since the free variables of $\phi$ are assumed to be all first-order.
    \item If $\Delta \sorts H K : \rho_2$ with $\Delta \sorts H : \rho_1 \to \rho_2$ and $\Delta \sorts K : \rho_1$, then we reason as follows.  First observe that:
    \[
      \begin{array}{rcl}
      \jembed_{\rho_2}(\mmng{HK}(\upmon(\alpha))) & = & \jembed_{\rho_2}(\mmng{H}(\upmon(\alpha))(\mmng{K}(\upmon(\alpha)))) \\
      &\subseteq&  \jembed_{\rho_2}(\mmng{H}(\upmon(\alpha))(\upmon_{\rho_1}(\jembed_{\rho_1}(\mmng{K}(\upmon(\alpha))))))
      \end{array}
    \]
    by definition and because $\jembed \circ \mmng{H}(\upmon(\alpha))$ is monotone and $\upmon \circ \jembed$ is inflationary.  Then:
    \[
      \begin{array}{rcl}
         \jembed_{\rho_2}(\mmng{H}(\upmon(\alpha))(\upmon_{\rho_1}(\jembed_{\rho_1}(\mmng{K}(\upmon(\alpha)))))) &
        = & \jembed_{\rho_1 \to \rho_2}(\mmng{H}(\upmon(\alpha)))(\jembed_{\rho_1}(\mmng{K}(\upmon(\alpha)))) \\
        &\subseteq& \jembed_{\rho_1 \to \rho_2}(\mmng{H}(\upmon(\alpha)))(\smng{K}(\alpha))
      \end{array}
    \]
    by definition, the induction hypothesis and because $\jembed(r)$ is monotone in general because $r$ is.  Finally:
    \[
      \jembed_{\rho_1 \to \rho_2}(\mmng{H}(\upmon(\alpha)))(\smng{K}(\alpha))
      \subseteq \smng{H}(\alpha)(\smng{K}(\alpha))
      = \smng{HK}(\alpha)
    \]
    by definition and the induction hypothesis.
  \item If $\Delta \sorts \abs{x}{H} : \rho_1 \to \rho_2$ and $\Delta,x:\!\rho_1 \sorts H : \rho_2$, we first observe that:
  \[
  \jembed_{\rho_1 \to \rho_2}(\mmng{\abs{x}{H}}(\upmon(\alpha)))(s)
  = \jembed_{\rho_2}(\mmng{\abs{x}{H}}(\upmon(\alpha))(\upmon_{\rho_1}(s)))
  = \jembed_{\rho_2}(\mmng{H}(\upmon(\alpha[x \mapsto s])))
  \]
  follows from the definitions.  Then we note that:
  \[
    \jembed_{\rho_2}(\mmng{H}(\upmon(\alpha[x \mapsto s])))
    \subseteq \smng{H}(\alpha[x \mapsto s]))
    = \smng{\abs{x}{H}}(\alpha)(s)
  \]
  follows by definition and the induction hypothesis.
  \item If $\Delta \sorts \vee : o \to o \to o$ or $\Delta \sorts \wedge : o \to o \to o$, the result holds by definition.
  \item If $\Delta \sorts \exists_\sigma : (\sigma \to o) \to o$ then:
  \[
    \jembed(\mmng{\exists}(\upmon(\alpha))) = \jembed(\mexistsfn) = \mexistsfn \circ \upmon \subseteq \existsfn = \smng{\exists}(\alpha)
  \]
  follows by definition and Lemma \ref{lem:exists-mngs}.
  \end{itemize}
  Showing the inclusion $\smng{G} \subseteq \iembed \circ \mmng{G} \circ \lomon$ is dual: just observe that whenever, in the above proof, an inclusion is justified by $\upmon \circ \jembed$ being inflationary (respectively $\jembed \circ \upmon$ deflationary), the corresponding reverse inclusion can be justified by noting that $\lomon \circ \iembed$ is deflationary (respectively $\iembed \circ \lomon$ inflationary).
\end{proof}

%% file: apx-ty-assignment.tex
\section{Supplementary Material for Section 5}

\subsection{Proof of Lemma \ref{lem:types-are-pideals}}

\begin{lemma}\label{lem:ty-interp-mono}
Let $\Delta \sorts T_1 \to T_2 :: \rho_1 \to \rho_2$ be a type and let $\alpha \in \mmng{\Delta}$ be a first order valuation.  Then $\rmmng{T_1 \to T_2}(\alpha)$ is monotone.
\end{lemma}
\begin{proof}
Let $t_1,t_2 \in \mmng{\rho_1}$ and consider the following two cases.
\begin{itemize}
\item If $t_1 \subseteq \rmmng{T_1}(\alpha)$, then $\rmmng{T_1 \to T_2}(\alpha)(t_1) = \rmmng{T_2}(\alpha)$.  Since $\rmmng{T_1 \to T_2}(\alpha)(t_2)$ is either $\rmmng{T_2}(\alpha)$ or $\top$, it follows that $\rmmng{T_1 \to T_2}(\alpha)(t_1) \subseteq \rmmng{T_1 \to T_2}(\alpha)(t_2)$.
\item Otherwise, $\rmmng{T_1 \to T_2}(\alpha)(t_1) = \top$ and it follows from the assumption that $t_2 \not\subseteq \rmmng{T_1}(\alpha)$.  Hence, also $\rmmng{T_1 \to T_2}(\alpha)(t_2) = \top$.
\end{itemize}
\end{proof}

\begin{lemn}[\ref{lem:types-are-pideals}]
For all types $\Delta \sorts T :: \rho$ and $\alpha \in \mmng{\Delta}$, $\mmng{T}(\alpha) = \pideal\rmmng{T}(\alpha)$.
\end{lemn}
\begin{proof}
  The proof is by induction on the derivation of the sorting judgement.
  \begin{itemize}
    \item When the judgement is of shape $\Delta \sorts \boolty{\phi} :: \boolsort$, clearly $\{0,\smng{\phi}(\alpha)\} = \pideal\rmmng{\phi}(\alpha)$.

    \item When the judgement is of shape $\Delta \sorts \dto{x}{\iota}{T} :: \iota \to \rho$ we first observe that, for all $n \in A_\iota$, it follows from the induction hypothesis that $\mmng{T}(\alpha[x \mapsto n]) = \pideal\rmmng{T}(\alpha[x \mapsto n])$.  To see downwards closure, let $t_2 \in \mmng{\dto{x}{\iota}{T}}$ and $t_1 \subseteq t_2$.  Then, for all $n \in A_\iota$, $t_2(n) \in \mmng{T}(\alpha[x \mapsto n])$.  Since $t_1(n) \subseteq t_2(n)$ and this set is downwards closed, it follows that $t_1(n)$ is also a member.  To see membership, $\rmmng{\dto{x}{\iota}{T}}(\alpha)(n) \in \mmng{T}(\alpha[x \mapsto n])$ so, by definition, $\rmmng{\dto{x}{\iota}{T}}(\alpha) \in \mmng{\dto{x}{\iota}{T}}(\alpha)$.  Finally, let $r \in \mmng{\dto{x}{\iota}{T}}(\alpha)$ and let $n \in A_\iota$.  Then it follows that $r(n) \in \mmng{T}(\alpha[x \mapsto n])$ so $r(n) \subseteq \rmmng{T}(\alpha[x \mapsto n])$ as required.

    \item When the judgement has shape $\Delta \sorts T_1 \to T_2 :: \rho_1 \to \rho_2$ we proceed as follows.  To see that $\mmng{T_1 \to T_2}(\alpha)$ is downwards closed, let $t_2$ be a member and $t_1 \subseteq t_2$.  Then let $t \in \mmng{T_1}(\alpha)$.  It follows that $t_2(t) \in \mmng{T_2}(\alpha)$ and $t_1(t) \subseteq t_2(t)$, so the result follows from the induction hypothesis.  To see membership, let $t \in \mmng{T_1}(\alpha)$.  It follows from the induction hypothesis that, therefore $t \subseteq \rmmng{T_1}(\alpha)$.  Hence $\rmmng{T_1 \to T_2}(\alpha)(t) = \rmmng{T_2}(\alpha)$ and it follows from the induction hypothesis that $\rmmng{T_2}(\alpha) \in \mmng{T_2}(\alpha)$, as required.  To see the extremal property, let $s \in \mmng{T_1 \to T_2}(\alpha)$ and let $t \in \mmng{\rho_1}$.  If $t \in \mmng{T_1}(\alpha)$, so that $t \subseteq \rmmng{T_1}(\alpha)$ follows from the induction hypothesis, then the fact that $s(t) \subseteq \rmmng{T_1 \to T_2}(\alpha)(t)$ follows from the fact that $s(t) \in \mmng{T_2}(\alpha)$, $\rmmng{T_1 \to T_2}(\alpha)(t) = \rmmng{T_2}(\alpha)$ and the induction hypothesis.  Otherwise, $t \notin \mmng{T_1}(\alpha)$ and it follows from the induction hypothesis that, therefore, $t \not\subseteq \rmmng{T_1}(\alpha)$.  Hence $\rmmng{T_1 \to T_2}(\alpha)(t) = \top_{\rho_2}$ and the result is immediate.
  \end{itemize}
\end{proof}

\subsection{Proof of Theorem \ref{thm:ty-ass-soundness}}

\begin{thmn}[\ref{thm:ty-ass-soundness}]
For all $\sorts \theta : \Delta$, $\sorts \Gamma :: \Delta$, $\Delta \sorts G : \sigma$ and $\Delta \sorts T, T_1, T_2 :: \sigma$, the following is true:
\begin{enumerate}[(i)]
  \item $\types T_1 \subtype T_2$ implies $\models T_1 \subtype T_2$
  \item $\Gamma \types G : T$ implies $\Gamma \models G : T$
  \item $\types P : \Gamma$ implies $\models P : \Gamma$
\end{enumerate}
\end{thmn}
\begin{proof}
The proof of the first claim is by induction on the derivation.
\begin{itemize}
  \item If the conclusion is $\types \intty \subtype \intty$ the result follows immediately.
  \item If the conclusion is $\types \boolty{\phi} \subtype \boolty{\psi}$, then necessarily, for all $\alpha \in \mmng{\Delta}$, $\alpha \models \phi \implies \psi$.  Then let $\alpha \in \mmng{\Delta}$.  It follows that $\mmng{\boolty{\phi}}(\alpha) = \{0,\mmng{\phi}(\alpha)\} \subseteq \{0,\mmng{\psi}(\alpha)\}$.
  \item If the conclusion is $\types \dto{x}{T_1}{T_2} \subtype \dto{y}{T_1'}{T_2'}$, let $\alpha \in \mmng{\Delta}$, $r \in \mmng{x:T_1 \to T_2}(\alpha)$ and $s \in \mmng{T_1'}$.  It follows from the induction hypothesis that
  \[
  \begin{array}{c}
  \mmng{T_1'}(\alpha) = \mmng{T_1'[z/y]}(\gamma) \subseteq \mmng{T_1[z/x]}(\beta) = \mmng{T_1}(\alpha) \\
  \mmng{T_2}(\alpha) = \mmng{T_2[z/x]}(\beta) \subseteq \mmng{T_2'[z/y]}(\gamma) = \mmng{T_2'}(\alpha)
  \end{array}
  \]
  where $\beta = (\alpha \setminus \{x \mapsto \alpha(x)\}) \cup \{z \mapsto \alpha(x)\}$ and $\gamma = (\alpha \setminus \{y \mapsto \alpha(y)\}) \cup \{z \mapsto \alpha(y)\}$.
  It follows that $s \in \mmng{T_1}(\alpha)$ and hence $r(s) \in \mmng{T_2[s/y]}(\alpha)$.  Finally, we observe that therefore $r(s) \in \mmng{T_2'[s/y]}(\alpha)$.
\end{itemize}
The proof of the second claim is by induction on the typing derivation.
\begin{itemize}
  \item When the judgement has shape $\Gamma_1, x:T, \Gamma_2 \types x : T$, assume $\alpha \models \Gamma$.  Then $\mmng{G}(\alpha) = \alpha(x)$.  From our assumption, $\alpha(x) \in \mmng{\Gamma(x)}$, i.e. $\alpha(x) \in \mmng{T}(\alpha)$ as required.
  \item When the judgement has shape $\Gamma \types \phi : \boolty{\phi}$, it follows immediately that $\mmng{\phi}(\alpha) \in \mmng{\boolty{\phi}}$.
  \item When the judgement has shape $\Gamma \types G \wedge H : \boolty{\phi \wedge \psi}$, let $\alpha \models \Gamma$.  It follows from the induction hypothesis that $\Gamma \models G : \boolty{\phi}$ and $\Gamma \models H : \boolty{\psi}$ and consequently, $\mmng{G}(\alpha) \in \mmng{\boolty{\phi}}$ and $\mmng{H}(\alpha) \in \mmng{\boolty{\phi}}$.  The result then follows from Lemma \ref{lem:boolty-meaning}.
  \item When the judgement has shape $\Gamma \types G \vee H : \boolty{\phi \vee \psi}$ the proof is analogous to the case above.
  \item When the judgement has shape $\Gamma \types \exists x.G : \boolty{\exists x.\phi}$ let $\alpha \models \Gamma$.  It follows from the induction hypothesis that $\Gamma,x:\intsort \models G : \boolty{\phi}$.  Hence, for any $n \in \mmng{\intsort}$, $\alpha[x \mapsto n] \models G \implies \phi$.  Assume $\alpha \models \exists x.G$, then there is some $n$ such that $\alpha[x \mapsto n] \models G$ which has the consequence that $\alpha[x \mapsto n] \models \phi$ and hence $\alpha \models \exists x.\phi$.  The result then follows from Lemma \ref{lem:boolty-meaning}.
  \item When the judgement has shape $\Gamma \types \term{G H} : T_2[H/x]$ and $\Gamma \types H : \intty$, let $\alpha \models \Gamma$.  It follows from the induction hypothesis that $\mmng{G}(\alpha) \in \mmng{\dto{x}{\intty}{T_2}}(\alpha)$ and $\mmng{H}(\alpha) \in \mmng{\intsort}(\alpha)$.  Hence, by definition, $\mmng{\term{G H}}(\alpha) \in \mmng{T_2}(\alpha[x \mapsto \mmng{H}(\alpha)]) = \mmng{T_2[H/x]}(\alpha)$.
  \item When the judgement has shape $\Gamma \types \term{G H} : T_2$ and $\Gamma \types H : T_1$ and $T_1$ is not $\intty$, let $\alpha \models \Gamma$.  It follows from the induction hypothesis that $\mmng{G}(\alpha) \in \mmng{T_1 \to T_2}$ and $\mmng{H} \in \mmng{T_1}$.  Hence, it follows by definition that $\mmng{\term{G H}}(\alpha) \in \mmng{T_2}(\alpha)$.
  \item When the judgement has shape $\Gamma \types \term{\abs{x:\intty}{G}} : \dto{x}{\intty}{T}$, let $\alpha \models \Gamma$.  It follows from the induction hypothesis that, for all $n \in \mmng{\intsort}$, $\mmng{G}(\alpha[x \mapsto n]) \in \mmng{T}(\alpha[x \mapsto n])$.  By definition $\mmng{\abs{x:\intty}{G}}(\alpha) \in \mmng{\dto{x}{\intty}{T}}(\alpha)$ just if, for all $n \in \mmng{\intsort}$, $\mmng{\abs{x:\intty}{G}}(\alpha)(n) \in \mmng{T}(\alpha[x \mapsto n])$.  The result then follows from the definition of $\mmng{\abs{x}{G}}(\alpha)(n)$ and the previous observation.
  \item When the judgement has shape $\Gamma \types \abs{x:\rho}{G} : T_1 \to T_2$ for $T_1 \neq \intty$, let $\alpha \models \Gamma$.  It follows from the induction hypothesis that, for all $r \in \mmng{T_1}(\alpha)$, then $\mmng{G}(\alpha[x \mapsto r]) \in \mmng{T_2}(\alpha[x \mapsto r])$.  Let $r \in \mmng{T_1}(\alpha)$, then $\mmng{\abs{x}{G}}(\alpha)(r) = \mmng{G}(\alpha[x \mapsto r])$.  It follows from the previous observation that this expression is an element of $\mmng{T_2}(\alpha[x \mapsto r])$, but $T_2$ cannot have any occurrence of relational variable $x$ since it is built out of constraint formulas.
\end{itemize}
In the third case we reason as follows.
Assume $\types P : \Gamma$ so that, necessarily, for all $x \in \dom(\Delta)$, $\Gamma \types P(x) : \Gamma(x)$.
Then it follows from part (ii) that, for each $x \in \dom(\Delta)$, and $\alpha$ satisfying $\Gamma$, $\mmng{P(x)}(\alpha) \in \mmng{\Gamma(x)}$.
It follows from Lemma \ref{lem:types-are-pideals} that $\rmmng{\Gamma}$ satisfies $\Gamma$.
Consequently, for all $x \in \dom(\Delta)$, $\mmng{P(x)}(\rmmng{\Gamma}) \in \mmng{\Gamma(x)}$, which is to say that $\mmng{P(x)}(\rmmng{\Gamma}) \subseteq \rmmng{\Gamma}(x)$.
Hence, $\rmmng{\Gamma}$ is a prefixpoint of $\mfunc_{P:\Delta}$.
It follows from the canonicity of $\mmng{P}$ that $\mmng{P} \subseteq \rmmng{\Gamma}$ and, by Lemma \ref{lem:types-are-pideals}, that therefore, for all $x \in \dom(\Delta)$, $\mmng{P}(x) \in \mmng{\Gamma(x)}$.

\end{proof}

\subsection{Proof of Lemma \ref{lem:G-defns}}

\begin{lemn}[\ref{lem:G-defns}]
For any closed type $\sorts T :: \rho$:
\begin{enumerate}[(i)]
  \item The goal term $\com(T)$ is a $G$-definition of the class $\{r \in \mmng{\rho} \mid r \notin \mmng{T}(\emptyset)\}$.
  \item The goal term $\lar(\falsetm)(T)$ is a $G$-definition of the relation $\rmmng{T}(\emptyset)$.
\end{enumerate}
\end{lemn}
\begin{proof}
We generalise the statement, proving for all $\Delta \sorts T :: \rho$, $\alpha \in \mmng{\Delta}$ and goal formulas $\Delta \sorts G : \boolsort$: $\mmng{\com(T)}(\alpha) = \{r \in \mmng{\rho} \mid r \notin \mmng{T}(\alpha)\}$ and, for all $\vv{d}$ of the appropriate sorts: $\mmng{\lar(G)(T)}(\alpha)(\vv{d}) = 1$ iff $\mmng{G}(\alpha)$ = 1 or $\rmmng{T}(\alpha)(\vv{d}) = 1$.
The proof is by induction on $\Delta \sorts T :: \rho$.
\begin{itemize}
  \item If the judgement is of shape $\Delta \sorts \boolty{\phi} :: o$ then we reason as follows.
  \begin{enumerate}[(i)]

    \item For all $b \in \mathbbm{2}$:
    \[
      \mmng{\com(T)}(\alpha)(b) = \mmng{x \wedge \neg\phi}(\alpha[x \mapsto b])
    \]
    This latter expression evaluates to $1$ iff $\alpha[x \mapsto b] \not\models x \implies \phi$ (by the variable convention we assume that $x$ does not occur in $\phi$).
    It follows from Lemma \ref{lem:boolty-meaning} that this is the case iff $\mmng{x}(\alpha[x \mapsto b]) = b \notin \mmng{T}(\alpha[x \mapsto b]) = \mmng{T}(\alpha)$.

    \item Also, $\mmng{\lar(G)(T)}(\alpha) = \mmng{G \vee \phi}(\alpha) = \orfn(\mmng{G}(\alpha))(\mmng{\phi}(\alpha))$.  This latter expression denotes $1$ iff $\mmng{G}(\alpha) = 1$ or $\mmng{\phi}(\alpha) = 1$.
  \end{enumerate}

  \item If the judgement is of shape $\Delta \sorts \dto{x}{\iota}{T'} :: \iota \to \rho$ then we reason as follows.
  \begin{enumerate}[(i)]
    \item For all $r \in \mmng{\iota \to \rho}$:
    \[
      \mmng{\com(T)}(\alpha)(r) = \mmng{\exists x\!\!:\!\!\iota.\,\com(T')(z\,x)}(\alpha[z \mapsto r])
    \]
    This latter expression denotes $1$ iff there is some $n \in A_\iota$ such that $\mmng{\com(T')(x\,y)}(\alpha[z \mapsto r][x \mapsto n]) = 1$ and this is true iff there is some $n$ such that $\mmng{\com(T')}(\alpha[x \mapsto n])(r(n)) = 1$.  It follows from the induction hypothesis that this is true iff there is some $n$ such that $r(n) \notin \mmng{T'}(\alpha[x \mapsto n])$.  This is exactly $r \notin \mmng{\dto{x}{\iota}{T'}}(\alpha)$.

    \item For all $n \in A_\iota$ and $\vv{d}$:
    \[
      \mmng{\lar(G)(T)}(\alpha)(n)(\vv{d}) = \mmng{\lar(G)(T')}(\alpha[x \mapsto n])(\vv{d})
    \]
    It follows from the induction hypothesis that this denotes $1$ iff either $\mmng{G}(\alpha[x \mapsto n]) = 1$ or $\rmmng{T'}(\alpha[x \mapsto n])(\vv{d}) = 1$.  We may assume, by the variable convention, that $x$ does not occur in $G$.  This latter expression denotes $1$ iff $\rmmng{\dto{x}{\iota}{T'}}(\alpha)$ by definition.
  \end{enumerate}

  \item If the judgement is of shape $\Delta \sorts T_1 \to T_2 :: \rho_1 \to \rho_2$ then we reason as follows.
  \begin{enumerate}[(i)]
    \item For all $r \in \mmng{\rho_1 \to \rho_2}$:
    \[
      \mmng{\com(T)}(\alpha)(r) = \mmng{\com(T_2)(x\,\lar(\falsetm)(T_1))}(\alpha[x \mapsto r])
    \]
    This latter expression is true iff $\mmng{\com(T_2)}(\alpha)(r(\mmng{\lar(\falsetm)(T_1)}(\alpha))) = 1$.  It follows from the induction hypothesis, part (ii), that this is the case iff\\ $\mmng{\com(T_2)}(\alpha)(r(\rmmng{T_1}(\alpha))) = 1$.  By the monotonicity of the operator and the fact that $\rmmng{T_1}(\alpha)$ is the largest element of $\mmng{T_1}(\alpha)$ (Lemma \ref{lem:types-are-pideals}), this is true iff there is some $s \in \mmng{T_1}(\alpha)$ such that $\mmng{\com(T_2)}(\alpha)(r(s)) = 1$.  It follows from the induction hypothesis, part (i), that this is true iff there is some $s \in \mmng{T_1}(\alpha)$ such that $r(s) \notin \mmng{T_2}(\alpha)$, which is exactly $r \notin \mmng{T_1 \to T_2}(\alpha)$.

    \item Fix $r \in \mmng{\rho_1}(\alpha)$ and $\vv{d}$ appropriate to the argument sorts of $\rho_2$.  Then:
    \[
      \mmng{\lar(G)(T)}(\alpha)(r)(\vv{d}) = \mmng{\lar(G \vee \com(T_1)\,x)(T_2)}(\alpha[x \mapsto r])(\vv{d})
    \]
    It follows from the induction hypothesis, part (ii), that this expression denotes $1$ iff (P1) $\mmng{G \vee \com(T_1)\,x}(\alpha[x \mapsto r]) = 1$ or (P2) $\rmmng{T_2}(\alpha)(\vv{d}) = 1$.
    It follows that this first possibility (P1) is true iff $\mmng{G}(\alpha) = 1$ or $\mmng{\com(T_1)}(\alpha)(r) = 1$.
    It follows from the induction hypothesis, part (i), that $\mmng{\com(T_1)}(\alpha)(r) = 1$ iff $r \notin \mmng{T_1}(\alpha)$.
    Consequently, (P1) or (P2) iff $\mmng{G}(\alpha) = 1$ or (Q1) $r \notin \mmng{T_1}(\alpha)$ or (Q2) $\rmmng{T_2}(\alpha)(\vv{d}) = 1$.
    We claim that (Q1) or (Q2) iff $\rmmng{T}(\alpha)(r)(\vv{d}) = 1$.
    In the backward direction, assume that (Q1) does not hold.  Then, by Lemma \ref{lem:types-are-pideals}, $r \subseteq \rmmng{T_1}(\alpha)$ and $\rmmng{T}(\alpha)(r)(\vv{d}) = 1$ implies $\rmmng{T_2}(\alpha)(r)(\vv{d}) = 1$ which is (Q2).
    In the forward direction we analyse the two cases.
    In case (Q1), by Lemma \ref{lem:types-are-pideals}, $r \not\subseteq \rmmng{T_1}(\alpha)$.  Then $\rmmng{T}(\alpha)(r)(\vv{d}) = \top_{\rho_2}(r)(\vv{d}) = 1$ by definition.
    Otherwise we assume (Q2) and not (Q1).
    Then, by Lemma \ref{lem:types-are-pideals}, $r \subseteq \rmmng{T_1}(\alpha)$, so $\rmmng{T_2}(\alpha)(\vv{d}) = 1$ implies $\rmmng{T}(\alpha)(r)(\vv{d}) = 1$ by definition.
  \end{enumerate}
\end{itemize}
\end{proof}

%% file: apx-inference.tex
\section{Type Inference}\label{sec:apx-inference}

In this appendix we give an algorithm, presented as a collection of syntax directed rules of inference, for determining the typability of a logic program over quantifier free integer linear arithmetic.
For convenience, we assume that the signature of the constraint language includes, for each sort $\intsort \to \cdots \to \intsort \to \boolsort$ of any arity, a countable supply of uninterpreted relation constants\footnote{We use relation constants rather than variables to ensure we stay within a first-order language.}, which will appear in the inference as constraint formulas refining the propositional sort.
The finite subset of these constants that are used in a given inference become the unknown relation symbols to be solved for in a first-order constrained Horn clause problem.

\begin{figure}

\begin{prooftree}
\AxiomC{}
\LeftLabel{\ISubBool}
\UnaryInfC{$\phi \implies \psi \infers \boolty{\phi} \subtype \boolty{\psi}$}
\end{prooftree}

\begin{prooftree}
\AxiomC{$C_1 \infers T_1' \subtype T_1$}
\AxiomC{$C_2 \infers T_2 \subtype T_2'$}
\LeftLabel{\ISubArrow}
\BinaryInfC{$C_1 \wedge C_2 \infers T_1 \to T_2 \subtype T_1' \to T_2'$}
\end{prooftree}

\begin{prooftree}
\AxiomC{$C \infers T[z/x] \subtype T'[z/y]$}
\LeftLabel{\ISubProd}
\RightLabel{$z$ fresh}
\UnaryInfC{$\forall z:\intsort.\,C \infers \dto{x}{\intty}{T} \subtype \dto{y}{\intty}{T'}$}
\end{prooftree}

\bigskip

\begin{prooftree}
\AxiomC{\phantom{Hello}}
\LeftLabel{\IVar}
\UnaryInfC{$\Gamma_1,\,x\!\!:T,\,\Gamma_2 \mm \truetm \infers x : T$}
\end{prooftree}

\begin{prooftree}
\AxiomC{\phantom{World}}
\LeftLabel{\IConst}
\UnaryInfC{$\Gamma \mm \truetm \infers \phi : \boolty{\phi}$}
\end{prooftree}

\begin{prooftree}
\AxiomC{$\Gamma \mm C \infers G : \dto{x}{\intsort}{T}$}
\LeftLabel{\IAppI}
\UnaryInfC{$\Gamma \mm C \infers \term{G N} : T[N/x]$}
\end{prooftree}

\begin{prooftree}
\AxiomC{$\Gamma \mm C_1 \infers G : T_1 \to T_2$}
\AxiomC{$\Gamma \mm C_2 \infers H : T_3$}
\AxiomC{$C_3 \infers T_3 \subtype T_1$}
\LeftLabel{\IAppR}
\TrinaryInfC{$\Gamma \mm C_1 \wedge C_2 \wedge C_3 \infers \term{G H} : T_2$}
\end{prooftree}

\begin{prooftree}
\AxiomC{$\Gamma,\,x:\intty \mm C \infers G : T$}
\LeftLabel{\IAbsI}
\UnaryInfC{$\Gamma \mm \forall x:\intsort.\, C \infers \abs{x\!\!:\!\!\intsort}{G} : \dto{x}{\intty}{T}$}
\end{prooftree}

\begin{prooftree}
\AxiomC{$T_1 = \freshty{\Gamma^\flat}{\sigma}$}
\AxiomC{$\Gamma,\,x:T_1 \mm C \infers G : T_2$}
\LeftLabel{\IAbsR}
\BinaryInfC{$\Gamma \mm C \infers \abs{x\!\!:\!\!\rho}{G} : T_1 \to T_2$}
\end{prooftree}

\begin{prooftree}
\AxiomC{$\Gamma \mm C_1 \infers G : \boolty{\phi_1}$}
\AxiomC{$\Gamma \mm C_2 \infers H : \boolty{\phi_2}$}
\LeftLabel{\IAnd}
\BinaryInfC{$\Gamma \mm C_1 \wedge C_2 \infers G \wedge H : \boolty{\phi_1 \wedge \phi_2}$}
\end{prooftree}

\begin{prooftree}
\AxiomC{$\Gamma \mm C_1 \infers G : \boolty{\phi_1}$}
\AxiomC{$\Gamma \mm C_2 \infers H : \boolty{\phi_2}$}
\LeftLabel{\IOr}
\BinaryInfC{$\Gamma \mm C_1 \wedge C_2 \infers G \vee H : \boolty{\phi_1 \vee \phi_2}$}
\end{prooftree}

\begin{prooftree}
\AxiomC{$\Gamma,\, x\!\!:\intty \mm C \infers G : \boolty{\phi}$}
\LeftLabel{\IExists}
\UnaryInfC{$\Gamma \mm \forall x:\intsort.\,C \infers \exists x\!\!:\intsort.\,G : \boolty{\exists x.\,\phi}$}
\end{prooftree}

\bigskip

\begin{prooftree}
\AxiomC{
  $\Gamma = \freshenv{\Delta}$
}
\noLine
\UnaryInfC{
  $C_x' \infers T_x \subtype \Gamma(x)$ (for each $x \in \dom(\Delta)$)
}
\noLine
\UnaryInfC{
  $\Gamma \mm C_x \infers P(x) : T_x$ (for each $x \in \dom(\Delta)$)
}
\LeftLabel{\IProg}
\UnaryInfC{$\bigwedge_{x \in \dom(\Delta)} (C_x \wedge C_x') \infers P : \Gamma$}
\end{prooftree}
\caption{Rules of inference.}\label{fig:inference-rules}
\end{figure}

The sort of a well formed refinement type is determined uniquely so we define the \emph{underlying sort} $T^\flat$ of a refinement type $T$ recursively:
\[
  \begin{array}{rcl}
  \boolty{s}^\flat &=& \boolsort \\
  (\dto{x}{\intsort}{T})^\flat &=& \intsort \to T^\flat \\
  ({T_1} \to {T_2})^\flat &=& T_1^\flat \to T_2^\flat \\
  \end{array}
\]
The underlying sort environment $\Gamma^\flat$ of a type environment $\Gamma$ is obtained recursively as follows:
\[
\begin{array}{rcl}
  \epsilon^\flat &=& \emptyset \\
  (\Gamma, x\!\!:T)^\flat &=& \Gamma^\flat, x\!\!:T^\flat \\
\end{array}
\]

\paragraph{Fresh relational variables.}
We suppose a function $\mathsf{freshRel}$ that, given a sort environment $\Delta$ and a second-order relational sort $\rho$, yields some term $\term{R x_1 \cdots{} x_k}$ with
\[
  R\!\!: \underbrace{\intsort \to \cdots{} \to \intsort}_{\text{$k$-times}} \to \rho
\]
a fresh, second-order relation constant, %
and the set $\makeset{x_1:\intsort, \ldots, x_k:\intsort}$ exactly the subset of $\Delta$ consisting of variables of integer sort.
We denote this situation by the judgement:
  \begin{prooftree}
  \AxiomC{}
  \RightLabel{$R$ fresh}
  \UnaryInfC{$\term{R x_1 \cdots{} x_k} = \freshrel{\Delta}{\rho}$}
  \end{prooftree}

\paragraph{Fresh types.}
Similarly, given a sort environment $\Delta$ and a sort $\sigma$, we suppose a function $\mathsf{freshTy}$ that %
yields a choice of type $T$ refining $\sigma$ and which is constructed from fresh relations as described above.
We will write a judgement of shape $T = \freshty{\Delta}{\sigma}$, which is defined by the following system:
\begin{prooftree}
\AxiomC{}
\UnaryInfC{$\intty = \freshty{\Delta}{\intsort}$}
\end{prooftree}
\begin{prooftree}
\AxiomC{$\phi = \freshrel{\Delta}{\boolsort}$}
\UnaryInfC{$\boolty{\phi} = \freshty{\Delta}{\boolsort}$}
\end{prooftree}
\begin{prooftree}
\AxiomC{$T_1 = \freshty{\Delta}{\rho_1}$}
\AxiomC{$T_2 = \freshty{\Delta}{\rho_2}$}
\BinaryInfC{$T_1 \to T_2 = \freshty{\Delta}{\rho_1 \to \rho_2}$}
\end{prooftree}
\begin{prooftree}
\AxiomC{$T = \freshty{\Delta,z:\intsort}{\rho}$}
\RightLabel{$z$ fresh}
\UnaryInfC{$\dto{z}{\intty}{T} = \freshty{\Delta}{\intsort \to \rho}$}
\end{prooftree}

\paragraph{Fresh type environments.}
We extend the judgement to sort environments $\Delta$ of shape $x_1\!\!:\sigma_1,\ldots,x_k\!\!:\sigma_k$ by the following rule:
\begin{prooftree}
\AxiomC{$T_1 = \freshty{\Delta}{\sigma_1} \quad \cdots \quad T_k = \freshty{\Delta}{\sigma_k}$}
\UnaryInfC{$x_1\!\!:T_1,\ldots,x_k\!\!:T_k = \freshenv{\Delta}$}
\end{prooftree}

\paragraph{Inference.}
Rather than giving a recursive procedure directly, we define three syntax-directed typing judgements, which have the following shape:
\[
C \infers T_1 \subtype T_2 \qquad\qquad \Gamma \mm C \infers G : T \qquad\qquad C \infers P : \Gamma
\]
for $\Gamma$ the type environment in which the type assignment is conducted, $C$ a first-order Horn constraint describing possible assignments to the uninterpreted relational constants, $G$ the term and $T$ the inferred type.

The three judgement forms are defined according to the rules in Figure \ref{fig:inference-rules}, where
we have generalised the judgements $\mathsf{freshRel}$ and $\mathsf{freshTy}$ to sets of sorts in the obvious way.  Since the system is syntax directed it can be read as an algorithm by regarding:

\noindent
\begin{center}
\begin{tabular}{r@{\hspace{10pt}}l}
$C \infers T_1 \subtype T_2$ & $T_1$ and $T_2$ as inputs and $C$ as output. \\
$\Gamma \mm C \infers G : T$ & $\Gamma$ and $G$ as inputs and $C$ and $T$ as outputs. \\
$C \infers P : \Gamma$ & $P$ as input and $C$ and $\Gamma$ as outputs.
\end{tabular}
\end{center}

Given an instance of the monotone problem $\abra{\Delta,P,G}$, let $C_1 \infers P : \Gamma$ and $\Gamma \mm C_2 \infers G : \boolty{\phi}$.
The type environment $\Gamma$ and the type $\boolty{\phi}$, which we think of as outputs, are  built from types containing formulas with occurrences of fresh relational constants.
Then there is a type environment $\Gamma'$ (without fresh relational constants) such that $\types P : \Gamma'$ and $\Gamma' \types G : \boolty{\falsetm}$ iff the first order system of constrained Horn clauses $\abra{\Delta',C_1 \wedge C_2 \wedge \phi \implies \falsetm}$ has a symbolic model, i.e. a solution that is expressible within quantifier free integer linear arithmetic.

%% file: apx-extension.tex
\section{Extension by type guards}\label{sec:extension}

It is worthwhile to consider the kinds of safety verification questions that can be posed in our formalisation.
In particular, since we are interested in verifying higher-order programs, it would be helpful to be able to state properties of higher-order functions.
In type-based approaches to verification, such as \cite{rondon-et-al-pldi2008,unno-et-al-popl2013}, it is possible to do this.  For example, one can state a typing judgement
asserting that some program expression is guaranteed to map evenness preserving functions to non-zero integers.

It is not immediate how to state a similar property using higher-order constrained Horn clauses.
Here, properties are specified by stating their negation as a goal formula.
The negation of the example property above has the form:
$
  \exists f\!\!:\intsort \to \intsort \to \boolsort.\, \exists x:\!\!\intsort.\, \mathsf{EvenPreserving}\ f \wedge x = 0 \wedge \term{G f x}
$,
with $G$ playing the role of the program expression,
and $\mathsf{EvenPreserving}$ is required to be a goal formula asserting that its argument relates even integers only to other even integers. That is, a logical formulation of the refinement type $\dto{x}{\intty}{\dto{y}{\intty}{\boolty{x \equiv 0 \mathrel{\mathsf{mod}} 2 \implies y \equiv 0 \mathrel{\mathsf{mod}} 2}}}$.
However, the existence of such a goal term is problematic, since evenness preservation of this kind is not a monotone property --- it is clearly satisfied by the empty relation, yet violated by many relations larger than it.

The aim of this section is to show that, despite this serious deficiency of expression, it is possible to state higher-type properties of the kind given above.
The key is to observe that, in the formulation of a type judgement as a goal, the properties expressed by refinement types will only occur guarding existentially quantified variables.
Hence, these ``type guards'' (such as the hypothetical $\mathsf{EvenPreserving}$ predicate above), serve only to limit the search space from which the existential witness is drawn, rather than to  test some quality of a particular, concrete individual.
We begin by making the restricted shape of these type guards precise then, in Section \ref{sec:elim-guards}, we show how monotonicity allows for such guards to be entirely eliminated.

\paragraph{Type guarded existentials.}
We extend the syntax of higher-order Horn clauses by adding a family of \emph{type guarded existential quantifier} constants $\exists_{T::\rho} : \rho \to \boolsort$, which are parametrised by a closed type $\sorts T :: \rho$.  We write $\exists x\!\!:\!\!T.\,M$ rather than $\exists_{T::\rho}(\abs{x}{M})$, omitting $\rho$ since it can be uniquely determined from $T$.

\paragraph{Interpretation.}
For the purposes of interpretation, we need to give two standard semantics to types that parallels their monotone semantics.
The definitions are completely analogous:
\[
\begin{array}{rcl}
  \smng{\Delta \sorts \boolty{\phi} :: \boolsort}(\alpha) &=& \{0, \smng{\Delta \types \phi : \boolsort}(\alpha) \} \\
  \smng{\Delta \sorts \dto{x}{\iota}{T} :: \iota \to \rho}(\alpha) &=& \{ r  \mid \forall n \in A_\iota.\; r(n) \in \smng{\Delta,x:\iota \sorts T :: \rho}(\alpha[x \mapsto n])\} \\
  \smng{\Delta \sorts T_1 \to T_2 :: \rho_1 \to \rho_2}(\alpha) &=& \{ f  \mid \forall r \in \smng{\Delta \types T_1 :: \rho_1}(\alpha).\, f(r) \in \smng{\Delta \sorts T_2 :: \rho_2}(\alpha)\} \\
\end{array}
\]
but now the set from which e.g. $f$ is drawn in the third clause is the set of all propositional functions $\smng{\rho_1 \to \rho_2}$.  Similarly for the relational semantics.
\[
\begin{array}{rcll}
  \rsmng{\Delta \types \boolty{\phi} :: \boolsort}(\alpha) &=& \smng{\Delta \types \phi : \boolsort}(\alpha) &\\
  \rsmng{\Delta \types \dto{x}{\iota}{T} :: \iota \to \rho}(\alpha)(n) &=& \rsmng{\Delta, x:\iota \sorts T :: \rho}(\alpha[x \mapsto n]) &\\
  \rsmng{\Delta \types T_1 \to T_2 :: \rho_1 \to \rho_2}(\alpha)(r) &=&
      \rsmng{\Delta \sorts T_2 :: \rho_2}(\alpha) & \text{if $r \subseteq \rsmng{\Delta \sorts T_1 :: \rho_1}(\alpha)$} \\
  \rsmng{\Delta \types T_1 \to T_2 :: \rho_1 \to \rho_2}(\alpha)(r) &=& \rsmng{ \epsilon \sorts \toptype{\rho_2} :: \rho_2}(\emptyset) & \text{otherwise} \\
\end{array}
\]
It follows that $\rsmng{T}(\alpha)$ is actually a monotone relation. An inspection of the proof of Lemma \ref{lem:types-are-pideals} shows that it carries over to the case of the standard semantics.

Having defined this semantics of types, the family of type guarded existential quantifiers can be interpreted in the standard semantics by the family of functions $\existsfn_{T::\rho} \in \smng{\rho} \To \mathbbm{2}$ and in the monotone semantics by the family of functions $\mexistsfn_{T::\rho} \in \mmng{\rho} \mTo \mathbbm{2}$ defined by
$\existsfn_{T::\rho}(s) = \max\{s(d) \mid d \in \smng{\sorts T :: \rho} \}$ and
$\mexistsfn_{T::\rho}(r) = \max\{r(d) \mid d \in \mmng{\sorts T :: \rho} \}$.

\paragraph{Type guarded higher-order Horn clauses}
We extend the higher-order Horn clause problem to incorporate type guarded existentials by adding to the grammar for constrained goal formulas: $G \Coloneqq \cdots \mid \exists x\!\!:\!\!T.\: G$.
Note that this subsumes the existential quantifier on relations completely as: $\smng{\exists_\rho} = \smng{\exists_{\toptype{\rho} :: \rho}}$, since the meaning $\smng{\toptype{\rho}}(\alpha)$ of the top type in any valuation $\alpha$ is the whole universe $\smng{\rho}$.
The extended problem is called the \emph{type-guarded, higher-order, constrained Horn clause problem}.
We extend logic programs in the same way, adding the possibility that constant $c$ can be $\exists_{T::\rho}$ in the side condition of the rule $\GCst$.
The resulting problem we name the \emph{type-guarded, monotone logic safety problem}.
The reduction given by Theorem \ref{thm:reduction-to-mono} still holds in this extended setting.
\begin{theorem}\label{thm:ty-guard-red-s-to-m}
Type-guarded, higher-order constrained Horn clause problem $\abra{\Delta,D,G}$ is solvable iff type-guarded, monotone logic safety problem $\abra{\Delta,P_D,G}$ is solvable.
\end{theorem}

\subsection{Elimination of type guards within monotone models}\label{sec:elim-guards}

We now show that type-guarded existentials can be eliminated under the monotone semantics, that is: any instance of the type-guarded monotone problem can be reduced to an instance of the ordinary monotone problem.
The obvious approach to this is to try to capture, using a formula, the class of relations defined by a closed type.
If $H$ is a goal term representing the class $\mmng{T}(\emptyset)$, then a type guarded existential formula $\exists x\!\!:\!\!T.\,G$ can be eliminated in favour of an ordinary existential formula $\exists x.\, H\,x \wedge G$.

\paragraph{Goal definability.}
We say that a relation $r \in \mmng{\rho}$ is \emph{goal term definable} (more briefly \emph{$G$-definable}) just if there exists a closed goal term $\types H : \rho$ and $\mmng{H}(\emptyset) = r$.  We say that a class of relations is $G$-definable just if the characteristic predicate of the class is $G$-definable.
However, having made this precise, we can now observe that it is not generally possible to capture the class of relations defined by a type.
\begin{lemma}\label{lem:non-definability}
The class of relations $\mmng{(\dto{x}{\intsort}{\boolty{x \equiv 0 \mathrel{\mathsf{mod}} 2}}) \to \boolty{\falsetm}}(\emptyset)$ is not $G$-definable.
\end{lemma}
\begin{proof}
This type, which we shall abbreviate by $T$, defines the predicate $\mmng{T}(\emptyset)$ which contains $r$ just if $r$ is false of every set of even integers (identifying the elements of $\mmng{\intsort \to \boolsort}$ with sets of integers).
However, this predicate $\mmng{T}(\emptyset)$ is not monotone and hence cannot be defined by a goal term.
Clearly, the empty relation $\bot \in \mmng{(\intsort \to \boolsort) \to \boolsort}$ is not an element of  $\mmng{T}(\emptyset)$, since it is false of every set of integers. However, $\bot \subseteq \top$, the universal relation, and yet $\top$ is not in $\mmng{T}(\emptyset)(\top)$ since it assigns true to all sets of integers.
\end{proof}
\noindent
This result makes it clear that although every $r \in \mmng{T}(\emptyset)$ is itself monotone, the set of relations $\mmng{T}(\emptyset)$, viewed as a monadic predicate, may not be monotone.

Despite the foregoing result, it is possible to eliminate type guards by representing their semantics logically.  Consider a type guarded formula $\exists x\!\!:\!\!T.\,G$.  We can think of $G$ as a monotone function of $x$, so, if $G$ is true of some relation $r \in \mmng{T}(\emptyset)$ then it will also be true of any larger relation in $\mmng{T}(\emptyset)$.
Consequently, $\exists x\!\!:\!\!T.\,G$ is true iff $G$ is true of the largest relation in $\mmng{T}(\emptyset)$, namely (following Lemma \ref{lem:types-are-pideals}) $\rmmng{T}(\emptyset)$.
Hence, we do not need to represent the whole class of relations $\mmng{T}(\emptyset)$ but only its largest member because, if $H$ is a $G$-definition of $\rmmng{T}(\emptyset)$, then we can eliminate $\exists x\!\!:\!\!T\,G$ in favour of $G[H/x]$, which is logically equivalent.

\paragraph{Definability of type complements and the relational semantics.}
In order to give a $G$-definition of the largest member of a type $\mmng{T :: \rho}(\emptyset)$, it is useful to be able to give a $G$-definition of the complement of the type, i.e. the class that contains $r \in \mmng{\rho}$ just if $r \notin \mmng{T}(\emptyset)$.
This is because the largest element of a type $T_1 \to T_2$ is a relation that, in particular, is true of its argument $s$ whenever $s$ is not an element of $\mmng{T_1}(\emptyset)$.
In fact, the notions of complement of $T$ and largest element of $T$ are here intertwined: To understand when a relation $r$ is not in the type $T_1 \to T_2$ you must understand when there exists a relation $s \in \mmng{T_1}(\emptyset)$ such that $r(s)$ is not in $\mmng{T_2}(\emptyset)$.
Since this condition is a  type-guarded existential, it follows from the above discussion that it can be re-expressed using the largest element of $\mmng{T_1}(\emptyset)$.
This leads to the definitions by mutual induction on type:
\[
  \begin{array}{ccc}
    \arraycolsep=1.4pt
    \begin{array}{rcl}
      \com(\boolty{\phi}) &=& \abs{z}{z \wedge \neg\phi} \\
      \com(\dto{x}{\iota}{T}) &=& \abs{z}{\exists  x\!:\!\iota.\,\com(T)\,(z\,x)} \\
      \com(T_1 \to T_2) &=& \abs{z}{\com(T_2)\,(z\,\lar(\falsetm)(T_1))}
    \end{array}
    &&
    \arraycolsep=1.4pt
    \begin{array}{rcl}
      \lar(G)(\boolty{\phi}) &=& G \vee \phi \\
      \lar(G)(\dto{x}{\iota}{T}) &=& \abs{z}{\lar(G)(T[z/x])} \\
      \lar(G)(T_1 \to T_2) &=& \abs{z}{\lar(G \vee{} \com(T_1)\,z)(T_2)}
    \end{array}
  \end{array}
\]
in which $\com(T)$ is a $G$-definition of the complement of the class $\mmng{T}(\emptyset)$.
The definition of $\lar$ is parametrised by a goal formula representing domain conditions that are accumulated during the analysis of function types.
It follows that $\lar(\falsetm)(T)$ is a $G$-definition of the largest element of the class $\mmng{T}(\emptyset)$.\\

\begin{lemn}[\ref{lem:G-defns}]
For any closed type $\sorts T :: \rho$:
\begin{enumerate}[(i)]
  \item The goal term $\com(T)$ is a $G$-definition of the class $\{r \in \mmng{\rho} \mid r \notin \mmng{T}(\emptyset)\}$.
  \item The goal term $\lar(\falsetm)(T)$ is a $G$-definition of the relation $\rmmng{T}(\emptyset)$.
\end{enumerate}
\end{lemn}
\begin{proof}
The following more general statement can be proven by a straightforward induction on $\Delta \sorts T :: \rho$. For all $\Delta \sorts T :: \rho$, $\alpha \in \mmng{\Delta}$ and goal formulas $\Delta \sorts G : \boolsort$, $\mmng{\com(T)}(\alpha) = \{r \in \mmng{\rho} \mid r \notin \mmng{T}(\alpha)\}$ and, for all $\vv{d}$ of the appropriate sorts: $\mmng{\lar(G)(T)}(\alpha)(\vv{d}) = 1$ iff $\mmng{G}(\alpha)$ = 1 or $\rmmng{T}(\alpha)(\vv{d}) = 1$.
\end{proof}

With a $G$-definition of the largest element satisfying a type, we have the necessary apparatus to eliminate type-guarded existentials.  To this end, for each type guarded goal term $G$, let $\mathsf{Elim}(G)$ be defined by:
\[
  \begin{array}{cc}
  \begin{array}{rcl}
    \mathsf{Elim}(x) &=& x \\
    \mathsf{Elim}(\abs{x\!\!:\!\!\iota}{G}) &=& \abs{x\!\!:\!\!\iota}{\mathsf{Elim}(G)} \\
    \mathsf{Elim}(\abs{x\!\!:\!\!\rho}{G}) &=& \abs{x\!\!:\!\!\rho}{\mathsf{Elim}(G)} \\
    \mathsf{Elim}(G_1\,N) &=& \mathsf{Elim}(G_1)\,N \\
    \mathsf{Elim}(G_1\,G_2) &=& \mathsf{Elim}(G_1)\,\mathsf{Elim}(G_2)
  \end{array}
  &
  \begin{array}{rcl}
    \mathsf{Elim}(\phi) &=& \phi \\
    \mathsf{Elim}(G_1 \vee G_2) &=& \mathsf{Elim}(G_1) \vee \mathsf{Elim}(G_2) \\
    \mathsf{Elim}(G_1 \wedge G_2) &=& \mathsf{Elim}(G_1) \wedge \mathsf{Elim}(G_2) \\
    \mathsf{Elim}(\exists x\!:\!T.\,G) &=& \mathsf{Elim}(G)[\lar(\falsetm)(T)/x] \\
    \mathsf{Elim}(\exists x\!:\!\iota.\,G) &=& \exists x\!:\!\iota.\,\mathsf{Elim}(G) \\
  \end{array}
  \end{array}
\]
and extend this to type guarded logic programs $\sorts P : \Delta$ by, for all $x \in \dom(\Delta)$, $\mathsf{Elim}(P)(x) = \mathsf{Elim}(P(x))$.
It is immediate that $\mathsf{Elim}(G)$ is a goal term not containing any type guarded existentials.

\begin{example}
  We return to the motivating example of this section.
  Consider the type $T$, written explicitly as $\dto{x}{\intty}{\dto{y}{\intty}{\boolty{x \equiv 0 \mathrel{\mathsf{mod}} 2 \implies y \equiv 0 \mathrel{\mathsf{mod}} 2}}}$. It follows that $\lar(\falsetm)(T) = \abs{xy}{x \equiv 0 \mathrel{\mathsf{mod}} 2 \implies y \equiv 0 \mathrel{\mathsf{mod}} 2}$
  and $\com(T) = \abs{z}{\exists x.\,\exists y.\, z x y \wedge x \equiv 0 \mathrel{\mathsf{mod}} 2 \wedge y \not\equiv 0 \mathrel{\mathsf{mod}} 2}$ (after a little simplification).
  It is possible to express that some goal term $H : (\intsort \to \intsort \to \boolsort) \to \intsort \to \boolsort$ is not an element of $\mmng{T \to \dto{z}{\iota}{\boolty{z \neq 0}}}(\emptyset)$ using the goal term $\com(T \to \dto{z}{\iota}{\boolty{z \neq 0}})\ H$.  This term asserts that $H$ violates the (relational formulation of the) property of mapping eveness preserving functions to non-zero integers and is equivalent to:
  \[
    \exists z.\, H\ (\abs{xy}{x \equiv 0 \mathrel{\mathsf{mod}} 2 \implies y \equiv 0 \mathrel{\mathsf{mod}} 2})\ z \wedge z = 0
  \]
  Here we use $\com$ to assert directly the negation of the property represented by the type.
  However, as mentioned in the discussion following the definition of $\com$, type guarded existentials are used implicitly.  For example, it is easy to see that this goal formula is logically equivalent to  $\mathsf{Elim}(\exists z_1\!\!:\!(T \to \dto{z}{\iota}{\boolty{z \neq 0}}).\, \exists z_2.\, H\ z_1 z_2 \wedge z_2 = 0)$.
\end{example}

As sketched above, due to monotonicity the (ordinary) goal term $\mathsf{Elim}(G)$ obained from type-guarded goal term $G$ is logically equivalent to $G$.
Consequently, there is no loss in performing the elimination: an instance of the type-guarded monotone problem is solvable iff the instance of the (ordinary) monotone problem obtained by elimination is solvable.
The proof follows the sketch outlined following Lemma \ref{lem:non-definability}.

\begin{theorem}\label{thm:ty-guard-elim-final}
Type guarded, higher-order constrained Horn clause problem $\abra{\Delta,D,G}$ is solvable iff monotone logic program safety problem $\abra{\Delta,\mathsf{Elim}(P_D),\mathsf{Elim}(G)}$ is solvable.
\end{theorem}

\noindent
Since every existential quantifier $\exists_\rho$ at relational sort can be viewed as an existential quantifier $\exists_{\top::\rho}$ guarded by the top type $\top_\rho$, it follows that, without loss of generality, one may assume that the latter problem contains only existential quantification over individuals.
This justifies the absence of relational sort existential quantifiers in the type system of Section \ref{sec:types}.